\documentclass[draft]{llncs}
\usepackage{amsfonts,amsmath,amssymb,color,latexsym,makeidx,times}
\definecolor{monblu}{cmyk}{0.80,0.00,0.00,0.70}
\definecolor{monmag}{cmyk}{0.00,1.00,0.00,0.00}
\newenvironment{proofclaim}{{\em Proof:}\/}{}
\def\N{\mathbb{N}}
\def\Axiom{\mathbf{A}}
\def\Rule{\mathbf{R}}
\def\IEL{\mathbf{IEL}}
\def\Fo{\mathbf{Fo}}
\def\Distrib{\mathbf{D}}
\def\ArtemovProtopopescu{\mathbf{AP}}
\def\FischerServi{\mathbf{FS}}
\def\Wijesekera{\mathbf{W}}
\def\doxastic{\mathbf{dox}}
\def\updown{\mathbf{ud}}
\def\refsym{\mathbf{rs}}
\def\transitive{\mathbf{tra}}
\def\partition{\mathbf{par}}
\def\all{\mathbf{all}}
\def\prestandard{\mathbf{pre}}
\def\standard{\mathbf{sta}}
\def\epistemic{\mathbf{epi}}
\def\At{\mathbf{At}}
\def\Ag{\mathbf{Ag}}
\def\IS{\mathbf{IS}}
\def\IK{\mathbf{IK}}
\def\S{\mathbf{S}}
\def\Log{\mathtt{Log}}
\def\IPL{\mathbf{IPL}}
\def\L{\mathbf{L}}
\def\Sf{\mathtt{sf}}
\begin{document}
\title{Intuitionistic modal logics:
\\
epistemic reasoning with distributed knowledge}
\author{Philippe Balbiani$^{\star\star}$}
\institute{
Toulouse Institute of Computer Science Research\footnote{Postal address: Institut de recherche en informatique de Toulouse, 118 route de Narbonne, 31062 Toulouse Cedex 9, France.
Email address: philippe.balbiani@irit.fr.}
\\
CNRS-INPT-UT, Toulouse, France
}
\maketitle
\section{Introduction}
In their intuitionistic epistemic logics, Artemov and Protopopescu~\cite{Artemov:2018,Artemov:Protopopescu:2016,Protopopescu:2015} propose a study of knowledge from an intuitionistic point of view.
Considering that knowledge is the product of verifications, they read formulas of the form ${\square}A$ as ``it is verified that $A$ holds intuitionistically''.
According to them, the reflection principle ${\square}p{\rightarrow}p$ should be rejected whereas two formulas characterize the properties of knowledge: $p{\rightarrow}{\square}p$ and ${\square}p{\rightarrow}{\neg}{\neg}p$.
In their article, among other things, Artemov and Protopopescu prove the completeness of their intuitionistic epistemic logics with respect to their appropriate relational semantics.
%
%
\\
\\
Modal logic is the study of two modal operators: a box $\square$ and a diamond $\lozenge$.
The above-men\-tioned articles being only about the box, it is natural to ask whether a diamond can be added to the intuitionistic epistemic logics put forward by Artemov and Protopopescu.\footnote{Here, the reader should remind that within the intuitionistic context, the box and the diamond are not interdefinable.}
In the relational semantics of intuitionistic modal logics, there are disagreements between the different interpretations of diamond~\cite{FischerServi:1984,Prenosil:2014,Wijesekera:1990}.
In this article, adopting a diamond \`a la P\v{r}enosil~\cite{Prenosil:2014}, we firstly consider multi-agent versions with distributed knowledge of the intuitionistic epistemic logics introduced in~\cite{Artemov:Protopopescu:2016}.
This leads us to accept for all groups $\alpha$ of agents, the formulas $p{\rightarrow}{\lbrack\alpha\rbrack}p$ and ${\lbrack\alpha\rbrack}p{\rightarrow}{\neg}{\neg}{\langle\alpha\rangle}p$ and leads us to the intuitionistic modal logics $\L_{\doxastic}$, $\L_{\epistemic}$, $\L_{\doxastic}^{\Distrib}$ and $\L_{\epistemic}^{\Distrib}$ axiomatically presented in Section~\ref{section:axiomatization:completeness}.
%
%
\\
\\
%
%
In this article, forgetting about the formulas $p{\rightarrow}{\square}p$ and ${\square}p{\rightarrow}{\neg}{\neg}p$ put forward by Artemov and Protopopescu, we also consider intuitionistic versions of the classical epistemic logics with distributed knowledge studied in~\cite{Halpern:Moses:1992,Wang:Agotnes:2020}.
This leads us to accept for all groups $\alpha$ of agents, the formulas ${\lbrack\alpha\rbrack}p{\rightarrow}p$, $p{\rightarrow}{\langle\alpha\rangle}p$, $p{\rightarrow}{\lbrack\alpha\rbrack}{\langle\alpha\rangle}p$ and ${\langle\alpha\rangle}{\lbrack\alpha\rbrack}p{\rightarrow}p$ and leads us to the intuitionistic modal logics $\L_{\partition}$ and $\L_{\partition}^{\Distrib}$ axiomatically presented in Section~\ref{section:axiomatization:completeness}.
%
%
\\
\\
Our main results are the proof of the completeness of $\L_{\doxastic}$, $\L_{\epistemic}$, $\L_{\partition}$, $\L_{\doxastic}^{\Distrib}$, $\L_{\epistemic}^{\Distrib}$ and $\L_{\partition}^{\Distrib}$ with respect to their appropriate relational semantics.
\section{Syntax}
\paragraph{Alphabet}
Let $\At$ be a countably infinite set (with typical members called {\em atoms}\/ and denoted $p$, $q$, etc).
Let $\Ag$ be a finite set (with typical members called {\em agents}\/ and denoted $a$, $b$, etc).\footnote{Our main results would still hold if $\Ag$ were infinite.
zzzzz}
Let ${\wp^{\star}}(\Ag)$ be the set (with typical members called {\em groups}\/ and denoted $\alpha$, $\beta$, etc) of all nonempty subsets of $\Ag$.\footnote{If $\Ag$ were infinite then we should decide whether groups could be arbitrary nonempty subsets of $\Ag$~---~in which case there would exists uncountably many groups~---~or groups should be finite nonempty subsets of $\Ag$~---~in which case there would exists countably many groups.
zzzzz}
%
%
%
%
\paragraph{$\ArtemovProtopopescu$-formulas}
Let $\Fo$ be the countably infinite set (with typical members called {\em $\ArtemovProtopopescu$-formulas,} or {\em formulas}\/ and denoted $A$, $B$, etc) of finite words over $\At{\cup}\{{\rightarrow},{\top},{\bot},{\vee},{\wedge}\}
$\linebreak$
{\cup}\{\lbrack\alpha\rbrack:\ \alpha{\in}\wp^{\star}(\Ag)\}{\cup}\{\langle\alpha\rangle:\ \alpha{\in}\wp^{\star}(\Ag)\}{\cup}\{(,)\}$ defined by\footnote{$\ArtemovProtopopescu$-formulas will be the inhabitants of the different intuitionistic modal logics that we will introduce in Section~\ref{section:axiomatization:completeness}.
Within the classical context, for all groups $\alpha$, $\lbrack\alpha\rbrack$ would be chosen as a primitive and for all $\ArtemovProtopopescu$-formulas $A$, ${\langle\alpha\rangle}A$ would be the abbreviation of $({\lbrack\alpha\rbrack}(A{\rightarrow}{\bot}){\rightarrow}{\bot})$, or $\langle\alpha\rangle$ would be chosen as a primitive and for all $\ArtemovProtopopescu$-formulas $A$, ${\lbrack\alpha\rbrack}A$ would be the abbreviation of $({\langle\alpha\rangle}(A{\rightarrow}{\bot}){\rightarrow}{\bot})$.
Within our intuitionistic context, for all groups $\alpha$, $\lbrack\alpha\rbrack$ and $\langle\alpha\rangle$ are not interdefinable and both are primitive.}
$$A\ {::=}\ p{\mid}(A{\rightarrow}A){\mid}{\top}{\mid}{\bot}{\mid}(A{\vee}A){\mid}(A{\wedge}A){\mid}{\lbrack\alpha\rbrack}A{\mid}{\langle\alpha\rangle}A$$
where $p$ ranges over $\At$ and $\alpha$ ranges over ${\wp^{\star}}(\Ag)$.\footnote{For all groups $\alpha$ and for all $\ArtemovProtopopescu$-formulas $A$, ${\lbrack\alpha\rbrack}A$ and ${\langle\alpha\rangle}A$ are respectively read ``$A$ is consequence of $\alpha$'s distributed knowledge'' and ``$A$ is compatible with $\alpha$'s distributed knowledge''.}
We follow the standard rules for omission of the parentheses.
For all $\ArtemovProtopopescu$-formulas $A,B$, we write ${\neg}A$ as the abbreviation of $A{\rightarrow}{\bot}$ and $A{\leftrightarrow}B$ as the abbreviation of $(A{\rightarrow}B){\wedge}(B{\rightarrow}A)$.
For all groups $\alpha$ and for all sets $\Gamma,\Delta$ of $\ArtemovProtopopescu$-formulas, let $\lbrack\alpha\rbrack\Gamma{=}\{A{\in}\Fo:\ \lbrack\alpha\rbrack A{\in}\Gamma\}$ and $\langle\alpha\rangle\Delta{=}\{\langle\alpha\rangle A{\in}
$\linebreak$
\Fo:\ A{\in}\Delta\}$.
\paragraph{Diamond-free $\ArtemovProtopopescu$-formulas}
An $\ArtemovProtopopescu$-formula $A$ is {\em diamond-free}\/ if for all groups $\alpha,\beta,
$\linebreak$
\gamma$,
\begin{itemize}
\item if the modal operators $\lbrack\alpha\rbrack$ and $\lbrack\beta\rbrack$ occur in $A$ then $\alpha{=}\beta$,
\item the modal operator $\langle\gamma\rangle$ does not occur in $A$.
\end{itemize}
Let $\Fo^{-}$ be the set of all diamond-free $\ArtemovProtopopescu$-formulas.
Let $\Sf:\ \Fo^{-}{\longrightarrow}\wp(\Fo^{-})$ be the function inductively defined as follows:
\begin{itemize}
\item $\Sf(p){=}\{p\}$,
\item $\Sf(A{\rightarrow}B){=}\{A{\rightarrow}B\}{\cup}\Sf(A){\cup}\Sf(B)$,
\item $\Sf({\top}){=}\{{\top}\}$,
\item $\Sf({\bot}){=}\{{\bot}\}$,
\item $\Sf(A{\vee}B){=}\{A{\vee}B\}{\cup}\Sf(A){\cup}\Sf(B)$,
\item $\Sf(A{\wedge}B){=}\{A{\wedge}B\}{\cup}\Sf(A){\cup}\Sf(B)$,
\item $\Sf({\lbrack\alpha\rbrack}A){=}\{{\lbrack\alpha\rbrack}A\}{\cup}\Sf(A)$.
\end{itemize}
\paragraph{$\square$-formulas}
Let $\Fo^{\square}$ be the countably infinite set (with typical members called {\em $\square$-formulas}\/ and denoted $A$, $B$, etc) of finite words over $\At{\cup}\{{\rightarrow},{\top},{\bot},{\vee},{\wedge},\square,(,)\}$ defined by\footnote{$\square$-formulas are the inhabitants of the intuitionistic epistemic logics $\IEL^{-}$ and $\IEL$ introduced by Artemov and Protopopescu~\cite{Artemov:2018,Artemov:Protopopescu:2016,Protopopescu:2015}.
Notice that Artemov and Protopopescu do not include the diamond in their language.}
$$A\ {::=}\ p{\mid}(A{\rightarrow}A){\mid}{\top}{\mid}{\bot}{\mid}(A{\vee}A){\mid}(A{\wedge}A){\mid}{\square}A$$
where $p$ ranges over $\At$.
We follow the standard rules for omission of the parentheses.
For all $\square$-formulas $A,B$, we write ${\neg}A$ as the abbreviation of $A{\rightarrow}{\bot}$ and $A{\leftrightarrow}B$ as the abbreviation of $(A{\rightarrow}B){\wedge}(B{\rightarrow}A)$.
Let $\tau:\ \Fo^{-}{\longrightarrow}\Fo^{\square}$ be the function inductively defined as follows:
\begin{itemize}
\item $\tau(p){=}p$,
\item $\tau(A{\rightarrow}B){=}\tau(A){\rightarrow}\tau(B)$,
\item $\tau({\top}){=}{\top}$,
\item $\tau({\bot}){=}{\bot}$,
\item $\tau(A{\vee}B){=}\tau(A){\vee}\tau(B)$,
\item $\tau(A{\wedge}B){=}\tau(A){\wedge}\tau(B)$,
\item $\tau({\lbrack\alpha\rbrack}A){=}{\square}\tau(A)$.
\end{itemize}
\section{Relational semantics}\label{section:relational:semantics}
\paragraph{Frames}
A {\em frame}\/ is a relational structure of the form $(W,{\leq},{R})$ where $W$ is a nonempty set (with typical members called {\em states}\/ and denoted $s$, $t$, etc), $\leq$ is a preorder on $W$ and ${R}\ :\ {\wp^{\star}}(\Ag){\longrightarrow}\wp(W{\times}W)$ is a function.
Let ${\mathcal C}_{\all}$ be the class of all frames.
\paragraph{Doxastic frames and epistemic frames}
A frame $(W,{\leq},{R})$ is {\em doxastic}\/ if for all groups $\alpha$ and for all $s,t{\in}W$, if $s{R(\alpha)}t$ then $s{\leq}t$.
A doxastic frame $(W,{\leq},{R})$ is {\em epistemic}\/ if for all groups $\alpha$ and for all $s{\in}W$, there exists $t{\in}W$ such that $s{\leq}{\circ}{R(\alpha)}t$.
Let ${\mathcal C}_{\doxastic}$ and ${\mathcal C}_{\epistemic}$ be respectively the class of all doxastic frames and the class of all epistemic frames.\footnote{Our main motivation for the introduction of doxastic frames and epistemic frames is coming from the articles of Artemov and Protopopescu~\cite{Artemov:2018,Artemov:Protopopescu:2016,Protopopescu:2015}.
With their intuitionistic epistemic logics $\IEL^{-}$ and $\IEL$, Artemov and Protopopescu propose a study of knowledge from an intuitionistic point of view.
Considering that knowledge is the product of verifications, they read $\square$-formulas of the form ${\square}A$ as ``it is verified that $A$ holds intuitionistically''.
According to them, the reflection principle ${\square}p{\rightarrow}p$ should be rejected whereas two $\square$-formulas characterize the properties of knowledge: $p{\rightarrow}{\square}p$ and ${\square}p{\rightarrow}{\neg}{\neg}p$.
In their articles, Artemov and Protopopescu prove $(\ast)$~the completeness of $\IEL^{-}$ with respect to the relational semantics determined by the class of all relational structures of the form $(W,{\leq},{R})$~---~called {\em $\IEL^{-}$-structures}\/~---~where $W$ is a nonempty set, $\leq$ is a preorder on $W$ and $R$ is a binary relation on $W$ such that $\mathbf{(i)}$~for all $s,t{\in}W$, if $s{R}t$ then $s{\leq}t$ and $\mathbf{(ii)}$~for all $s,t{\in}W$, if $s{\leq}{\circ}{R}t$ then $s{R}t$ and $(\ast\ast)$~the completeness of $\IEL$ with respect to the relational semantics determined by the class of all $\IEL^{-}$-structures $(W,{\leq},{R})$~---~called {\em $\IEL$-structures}\/~---~where in addition, $\mathbf{(iii)}$~for all $s{\in}W$, there exists $t{\in}W$ such that $s{R}t$.
When a relational structure $(W,{\leq},{R})$ of that form is equipped with an intuitionistic valuation $V\ :\ \At{\longrightarrow}\wp(W)$, Artemov and Protopopescu inductively define the satisfiability of $\square$-formulas as follows:
\begin{itemize}
\item $s{\models}p$ if and only if $s{\in}V(p)$,
\item $s{\models}A{\rightarrow}B$ if and only if for all $t{\in}W$, if $s{\leq}t$ then $t{\not\models}A$, or $t{\models}B$,
\item $s{\models}{\top}$,
\item $s{\not\models}{\bot}$,
\item $s{\models}A{\vee}B$ if and only if $s{\models}A$, or $s{\models}B$,
\item $s{\models}A{\wedge}B$ if and only if $s{\models}A$ and $s{\models}B$,
\item $s{\models}{\square}A$ if and only if for all $t{\in}W$, if $s{R}t$ then $t{\models}A$.
\end{itemize}
See next paragraphs for details about valuations and satisfiability.}
\paragraph{Reflexive frames, symmetric frames, transitive frames and partitions}
A frame $(W,{\leq},
$\linebreak$
{R})$ is {\em reflexive}\/ if for all groups $\alpha$ and for all $s{\in}W$, $s{R(\alpha)}s$.
A frame $(W,{\leq},{R})$ is {\em symmetric}\/ if for all groups $\alpha$ and for all $s,t{\in}W$, if $s{R(\alpha)}t$ then $t{R(\alpha)}s$.
A frame $(W,{\leq},{R})$ is {\em transitive}\/ if for all groups $\alpha$ and for all $s,t,u{\in}W$, if $s{R(\alpha)}t$ and $t{R(\alpha)}u$ then $s{R(\alpha)}u$.
Let ${\mathcal C}_{\refsym}$ and ${\mathcal C}_{\transitive}$ be respectively the class of all reflexive and symmetric frames and the class of all transitive frames.
A reflexive, symmetric and transitive frame is called a {\em partition.}
Let ${\mathcal C}_{\partition}$ be the class of all partitions.\footnote{Our main motivation for the introduction of reflexive frames, symmetric frames, transitive frames and partitions is coming from classical epistemic logics where knowledge is studied from a classical point of view~---~see~\cite{vanDitmarsch:et:al:2008,Fagin:et:al:1995,Meyer:VanDerHoeke:1995}.
In this setting, $\square$-formulas of the form ${\square}A$ are read as ``$A$ holds in every situation which is indiscernible from the current situation''.
As a result, the reflection principle ${\square}p{\rightarrow}p$ should be accepted as well as two $\square$-formulas characterizing the introspective properties of knowledge: ${\square}p{\rightarrow}{\square}{\square}p$ and ${\neg}{\square}p{\rightarrow}{\square}{\neg}{\square}p$.
In the above-mentioned literature, the completeness of classical epistemic logics is proved with respect to the relational semantics determined by relational structures of the form $(W,{R})$ where $W$ is a nonempty set and $R$ is an equivalence relation on $W$.
When a relational structure $(W,{R})$ of that form is equipped with a classical valuation $V\ :\ \At{\longrightarrow}\wp(W)$, the satisfiability of $\square$-formulas is inductively defined as follows:
\begin{itemize}
\item $s{\models}p$ if and only if $s{\in}V(p)$,
\item $s{\models}A{\rightarrow}B$ if and only if $s{\not\models}A$, or $s{\models}B$,
\item $s{\models}{\top}$,
\item $s{\not\models}{\bot}$,
\item $s{\models}A{\vee}B$ if and only if $s{\models}A$, or $s{\models}B$,
\item $s{\models}A{\wedge}B$ if and only if $s{\models}A$ and $s{\models}B$,
\item $s{\models}{\square}A$ if and only if for all $t{\in}W$, if $s{R}t$ then $t{\models}A$.
\end{itemize}
See next paragraphs for details about valuations and satisfiability.}
\paragraph{Up and down reflexive frames and up and down symmetric frames}
A frame $(W,{\leq},{R})$ is {\em up and down reflexive}\/ if for all groups $\alpha$ and for all $s{\in}W$, $s{{\leq}{\circ}{R(\alpha)}{\circ}{\leq}}s$ and $s{{\geq}{\circ}{R(\alpha)}{\circ}}
$\linebreak$
{{\geq}}s$.
A frame $(W,{\leq},{R})$ is {\em up and down symmetric}\/ if for all groups $\alpha$ and for all $s,t{\in}W$, if $s{R(\alpha)}t$ then $t{{\leq}{\circ}{R(\alpha)}{\circ}{\leq}}s$ and $t{{\geq}{\circ}{R(\alpha)}{\circ}{\geq}}s$.
Let ${\mathcal C}_{\updown}$ be the class of all up and down reflexive and up and down symmetric frames.\footnote{Obviously, every reflexive and symmetric frame is up and down reflexive and up and down symmetric.
Moreover, for all up and down reflexive and up and down symmetric frames $(W,{\leq},{R})$, the frame $(W^{\prime},{\leq^{\prime}},{R^{\prime}})$ where $W^{\prime}{=}W$, ${\leq^{\prime}}{=}{\leq}$ and ${R^{\prime}}{=}({\leq}{\circ}{R}{\circ}{\leq}){\cap}({\geq}{\circ}{R}{\circ}{\geq})$ is reflexive and symmetric.
This is our main motivation for the introduction of up and down reflexive and up and down symmetric frames.}
\paragraph{Prestandard frames and standard frames}
A frame $(W,{\leq},{R})$ is {\em prestandard}\/ if for all groups $\alpha,\beta$, ${R(\alpha{\cup}\beta)}{\subseteq}{R(\alpha)}{\cap}{R(\beta)}$.
A prestandard frame $(W,{\leq},{R})$ is {\em standard}\/ if for all groups $\alpha,\beta$, ${R(\alpha{\cup}\beta)}{=}{R(\alpha)}{\cap}{R(\beta)}$.
For all classes ${\mathcal C}$ of frames, let ${\mathcal C}^{\prestandard}$ and ${\mathcal C}^{\standard}$ be respectively the class of all prestandard frames in ${\mathcal C}$ and the class of all standard frames in ${\mathcal C}$.\footnote{As in the relational semantics of classical epistemic logics with distributed knowledge, we are introducing a relational semantics of intuitionistic epistemic logics with distributed knowledge where for all groups $\alpha$, the modal operators $\lbrack\alpha\rbrack$ and $\langle\alpha\rangle$ may be interpreted by means of the intersection of all accessibility relations associated to the agents in $\alpha$.
Notice that if $(W,{\leq},{R})$ is prestandard then for all groups $\alpha$, ${R(\alpha)}{\subseteq}{\bigcap}\{{R(\{a\})}:\ a{\in}\alpha\}$ whereas if $(W,{\leq},{R})$ is standard then for all groups $\alpha$, ${R(\alpha)}{=}{\bigcap}\{{R(\{a\})}:\ a{\in}\alpha\}$.}
\paragraph{Valuations and models}
For all frames $(W,{\leq},{R})$, a subset $U$ of $W$ is {\em closed}\/ if for all $s,t{\in}W$, if $s{\in}U$ and $s{\leq}t$ then $t{\in}U$.
An {\em intuitionistic valuation on a frame $(W,{\leq},{R})$,} or a {\em valuation on a frame $(W,{\leq},{R})$}\/ is a function $V\ :\ \At{\longrightarrow}\wp(W)$ such that for all atoms $p$, $V(p)$ is closed.\footnote{This {\em heredity property}\/ is standard in the intuitionistic setting.
It means that for all frames $(W,{\leq},{R})$, for all valuations $V$ on $(W,{\leq},{R})$, for all atoms $p$ and for all $s,t{\in}W$, if $p$ is true at $s$ and $t$ is in the $\leq$-future of $s$, $p$ is also true at $t$, i.e. we do not lose information when we go from $s$ to $t$.}
A {\em model based on the frame $(W,{\leq},{R})$}\/ is a structure of the form $(W,{\leq},{R},V)$ where $V\ :\ \At{\longrightarrow}\wp(W)$ is a valuation on $(W,{\leq},{R})$.
\paragraph{Satisfiability}
With respect to a model $(W,{\leq},{R},V)$, for all $s{\in}W$ and for all formulas $A$, the {\em satisfiability of $A$ at $s$ in $(W,{\leq},{R},V)$}\/ (in symbols $s{\models}A$) is inductively defined as follows:\footnote{In the following truth conditions, the reader will immediately notice that all connectives but $\rightarrow$ have their duals: $\top$ vs $\bot$, $\vee$ vs $\wedge$ and for all groups $\alpha$, $\lbrack\alpha\rbrack$ vs $\langle\alpha\rangle$.
This suggests the future study of intuitionistic epistemic logics with dual implication $\leftarrow$ interpreted as follows in all models $(W,{\leq},{R},V)$ and at all $s{\in}W$: $s{\models}A{\leftarrow}B$ if and only if there exists $t{\in}W$ such that $s{\geq}t$, $t{\not\models}A$ and $t{\models}B$.
Such dual implication has been developed by Rauszer~\cite{Rauszer:1977,Rauszer:1980}.}
\begin{itemize}
\item $s{\models}p$ if and only if $s{\in}V(p)$,
\item $s{\models}A{\rightarrow}B$ if and only if for all $t{\in}W$, if $s{\leq}t$ then $t{\not\models}A$, or $t{\models}B$,
\item $s{\models}{\top}$,
\item $s{\not\models}{\bot}$,
\item $s{\models}A{\vee}B$ if and only if $s{\models}A$, or $s{\models}B$,
\item $s{\models}A{\wedge}B$ if and only if $s{\models}A$ and $s{\models}B$,
\item $s{\models}{\lbrack\alpha\rbrack}A$ if and only if for all $t{\in}W$, if $s{\leq}{\circ}{R(\alpha)}t$ then $t{\models}A$,\footnote{The mono-agent version of this truth condition~---~$s{\models}{\square}A$ if and only if for all $t{\in}W$, if $s{\leq}{\circ}{R}t$ then $t{\models}A$~---~is used in many articles~\cite{FischerServi:1984,Prenosil:2014,Wijesekera:1990}.}
\item $s{\models}{\langle\alpha\rangle}A$ if and only if there exists $t{\in}W$ such that $s{\geq}{\circ}{R(\alpha)}t$ and $t{\models}A$.\footnote{The mono-agent version of this truth condition~---~$s{\models}{\lozenge}A$ if and only if there exists $t{\in}W$ such that $s{\geq}{\circ}{R}t$ and $t{\models}A$~---~is developed by P\v{r}enosil~\cite{Prenosil:2014} and mentioned by Simpson~\cite[Page~$49$]{Simpson:1994}.
See also~\cite{Balbiani:Gencer:APAL:2024}.
This truth condition should be compared to the truth conditions that Fischer Servi~\cite{FischerServi:1984} and Wijesekera~\cite{Wijesekera:1990} would have developed if they were using our syntax: $s{\models_{\FischerServi}}{\langle\alpha\rangle}A$ if and only if there exists $t{\in}W$ such that $s{R(\alpha)}t$ and $t{\models_{\FischerServi}}A$~; $s{\models_{\Wijesekera}}{\langle\alpha\rangle}A$ if and only if for all $t{\in}W$, if $s{\leq}t$ then there exists $u{\in}W$ such that $t{R(\alpha)}u$ and $u{\models_{\Wijesekera}}A$.
The definition of the satisfiability of formulas that Fischer Servi would have considered necessitates to restrict the discussion to the class of all {\em forward confluent frames,} i.e. the class of all frames $(W,{\leq},{R})$ such that for all groups $\alpha$, ${\geq}{\circ}{R(\alpha)}{\subseteq}{R(\alpha)}{\circ}{\geq}$, otherwise the heredity property described in Proposition~\ref{proposition:heredite} would not hold.
The definition of the satisfiability of formulas that Wijesekera would have considered does not necessitate to restrict the discussion to a specific class of frames.
The reader may easily verify that in the class of all forward confluent frames, the definition of the satisfiability of formulas that Fischer Servi would have used, the definition of the satisfiability of formulas that Wijesekera would have used and our definition of the satisfiability of formulas are equivalent.}
\end{itemize}
\begin{proposition}\label{proposition:heredite}
Let $(W,{\leq},{R},V)$ be a model.
For all formulas $A$ and for all $s,t{\in}W$, if $s{\models}A$ and $s{\leq}t$ then $t{\models}A$.\footnote{This {\em heredity property}\/ is of course a consequence of the heredity property imposed on the intuitionistic valuation $V$.
It is standard in the intuitionistic setting.}
\end{proposition}
\begin{proof}
By induction on $A$.
\medskip
\end{proof}
\paragraph{Truth and validity}
A formula $A$ is {\em true in a model $(W,{\leq},{R},V)$}\/ (in symbols $(W,{\leq},{R},
$\linebreak$
V){\models}A$) if for all $s{\in}W$, $s{\models}A$.
A formula $A$ is {\em valid in a frame $(W,{\leq},{R})$}\/ (in symbols $(W,{\leq},{R}){\models}A$) if for all models $(W,{\leq},{R},V)$ based on $(W,{\leq},{R})$, $(W,{\leq},{R},V){\models}A$.
A formula $A$ is {\em valid on a class ${\mathcal C}$ of frames}\/ (in symbols ${\mathcal C}{\models}A$) if for all frames $(W,{\leq},{R})$ in ${\mathcal C}$, $(W,{\leq},{R}){\models}A$.
For all classes ${\mathcal C}$ of frames, let $\Log({\mathcal C}){=}\{A{\in}\Fo\ :\ {\mathcal C}{\models}A\}$ be the {\em logic of ${\mathcal C}$.}
%
%
%
%
\section{Results about validities}
The following propositions will be used to show the soundness~---~with respect to their respective relational semantics~---~of the different intuitionistic modal logics that we will introduce in Section~\ref{section:axiomatization:completeness}.
\begin{proposition}\label{proposition:validity:of:axioms}
For all groups $\alpha,\beta$, the following formulas are valid on any class of frames:\footnote{The mono-agent versions of some of these formulas~---~${\square}p{\wedge}{\square}q{\rightarrow}{\square}(p{\wedge}q)$, ${\lozenge}(p{\vee}q){\rightarrow}{\lozenge}p{\vee}{\lozenge}q$, ${\square}{\top}$ and ${\neg}{\lozenge}{\bot}$~---~have been already considered in the above-mentioned literature about intuitionistic modal logics.
The mono-agent version of the $5$th formula~---~${\square}(p{\vee}q){\rightarrow}(({\lozenge}p{\rightarrow}{\square}q){\rightarrow}{\square}q)$~---~has been firstly considered in~\cite{Balbiani:et:al:2024:CSL,Balbiani:et:al:2024:IJCAR}.}
\begin{itemize}
\item ${\lbrack\alpha\rbrack}p{\wedge}{\lbrack\alpha\rbrack}q{\rightarrow}{\lbrack\alpha\rbrack}(p{\wedge}q)$,
\item ${\langle\alpha\rangle}(p{\vee}q){\rightarrow}{\langle\alpha\rangle}p{\vee}{\langle\alpha\rangle}q$,
\item ${\lbrack\alpha\rbrack}{\top}$,
\item ${\neg}{\langle\alpha\rangle}{\bot}$,
\item ${\lbrack\alpha\rbrack}(p{\vee}q){\rightarrow}(({\langle\alpha\rangle}p{\rightarrow}{\lbrack\alpha\rbrack}q){\rightarrow}{\lbrack\alpha\rbrack}q)$.
\end{itemize}
\end{proposition}
\begin{proof}
We only consider the $5$th formula.
If ${\lbrack\alpha\rbrack}(p{\vee}q){\rightarrow}(({\langle\alpha\rangle}p{\rightarrow}{\lbrack\alpha\rbrack}q){\rightarrow}{\lbrack\alpha\rbrack}q)$ is not valid in a class ${\mathcal C}$ of frames, there exists formulas $A,B$ and there exists a frame $(W,{\leq},
$\linebreak$
{R})$ in ${\mathcal C}$ such that $(W,{\leq},{R}){\not\models}{\lbrack\alpha\rbrack}(A{\vee}B){\rightarrow}(({\langle\alpha\rangle}A{\rightarrow}{\lbrack\alpha\rbrack}B){\rightarrow}{\lbrack\alpha\rbrack}B)$.
Hence, there exists a model $(W,{\leq},{R},V)$ based on $(W,{\leq},{R})$ such that $(W,{\leq},{R},V){\not\models}{\lbrack\alpha\rbrack}(A{\vee}B){\rightarrow}
$\linebreak$
(({\langle\alpha\rangle}A{\rightarrow}{\lbrack\alpha\rbrack}B){\rightarrow}{\lbrack\alpha\rbrack}B)$.
Thus, there exists $s{\in}W$ such that $s{\not\models}{\lbrack\alpha\rbrack}(A{\vee}B){\rightarrow}(({\langle\alpha\rangle}A{\rightarrow}
$\linebreak$
{\lbrack\alpha\rbrack}B){\rightarrow}{\lbrack\alpha\rbrack}B)$.
Consequently, there exists $t{\in}W$ such that $s{\leq}t$, $t{\models}{\lbrack\alpha\rbrack}(A{\vee}B)$ and $t{\not\models}
$\linebreak$
({\langle\alpha\rangle}A{\rightarrow}{\lbrack\alpha\rbrack}B){\rightarrow}{\lbrack\alpha\rbrack}B$.
Hence, there exists $u{\in}W$ such that $t{\leq}u$, $u{\models}{\langle\alpha\rangle}A{\rightarrow}{\lbrack\alpha\rbrack}B$ and $u{\not\models}{\lbrack\alpha\rbrack}B$.
Thus, there exists $v{\in}W$ such that $u{\leq}{\circ}{R(\alpha)}v$ and $v{\not\models}B$.
Since $t{\leq}u$, then $t{\leq}{\circ}{R(\alpha)}v$.
Since $t{\models}{\lbrack\alpha\rbrack}(A{\vee}B)$, then $v{\models}A{\vee}B$.
Consequently, $v{\models}A$, or $v{\models}B$.
Since $v{\not\models}B$, then $v{\models}A$.
Since $u{\leq}{\circ}{R(\alpha)}v$, then there exists $w{\in}W$ such that $u{\leq}w$ and $w{R(\alpha)}
$\linebreak$
v$.
Since $v{\models}A$, then $w{\models}{\langle\alpha\rangle}A$.
Since $u{\models}{\langle\alpha\rangle}A{\rightarrow}{\lbrack\alpha\rbrack}B$ and $u{\leq}w$, then $w{\models}{\lbrack\alpha\rbrack}B$.
Since $w{R(\alpha)}v$, then $v{\models}B$: a contradiction.
\medskip
\end{proof}
\begin{proposition}\label{proposition:validity:of:rules}
For all groups $\alpha$, the following rules preserve validity on any class of frames:\footnote{The mono-agent versions of some of these rules~---~$\frac{p{\rightarrow}q}{{\square}p{\rightarrow}{\square}q}$ and $\frac{p{\rightarrow}q}{{\lozenge}p{\rightarrow}{\lozenge}q}$~---~have been already considered in the above-mentioned literature about intuitionistic modal logics.
A variant of the mono-agent version of the $3$rd rule~---~$\frac{{\lozenge}p{\rightarrow}q{\vee}{\square}(p{\rightarrow}r)}{{\lozenge}p{\rightarrow}q{\vee}{\lozenge}r}$~---~has been firstly considered in~\cite{Prenosil:2014}.
See also~\cite{Balbiani:Gencer:APAL:2024}.}
\begin{itemize}
\item $\frac{p{\rightarrow}q}{{\lbrack\alpha\rbrack}p{\rightarrow}{\lbrack\alpha\rbrack}q}$,
\item $\frac{p{\rightarrow}q}{{\langle\alpha\rangle}p{\rightarrow}{\langle\alpha\rangle}q}$,
\item $\frac{{\langle\alpha\rangle}p{\rightarrow}q{\vee}{\lbrack\alpha\rbrack}(p{\rightarrow}r)}{{\langle\alpha\rangle}p{\rightarrow}q{\vee}{\langle\alpha\rangle}r}$.
\end{itemize}
\end{proposition}
\begin{proof}
We only consider the $3$rd rule.
If $\frac{{\langle\alpha\rangle}p{\rightarrow}q{\vee}{\lbrack\alpha\rbrack}(p{\rightarrow}r)}{{\langle\alpha\rangle}p{\rightarrow}q{\vee}{\langle\alpha\rangle}r}$ does not preserve validity on a class ${\mathcal C}$ of frames, there exists formulas $A,B,C$ and there exists a frame $(W,{\leq},{R})$ in ${\mathcal C}$ such that $(W,{\leq},{R}){\models}{\langle\alpha\rangle}A{\rightarrow}B{\vee}{\lbrack\alpha\rbrack}(A{\rightarrow}C)$ and $(W,{\leq},{R}){\not\models}{\langle\alpha\rangle}A{\rightarrow}
$\linebreak$
B{\vee}{\langle\alpha\rangle}C$.
Hence, there exists a model $(W,{\leq},{R},V)$ based on $(W,{\leq},{R})$ such that $(W,{\leq},{R},V){\not\models}{\langle\alpha\rangle}A{\rightarrow}B{\vee}{\langle\alpha\rangle}C$.
Thus, there exists $s{\in}W$ such that $s{\not\models}{\langle\alpha\rangle}A{\rightarrow}B{\vee}
$\linebreak$
{\langle\alpha\rangle}C$.
Consequently, there exists $t{\in}W$ such that $s{\leq}t$, $t{\models}{\langle\alpha\rangle}A$ and $t{\not\models}B{\vee}{\langle\alpha\rangle}C$.
Hence, there exists $u{\in}W$ such that $t{\geq}{\circ}{R(\alpha)}u$ and $u{\models}A$.
Thus, there exists $v{\in}W$ such that $t{\geq}v$ and $v{R(\alpha)}u$.
Since $u{\models}A$, then $v{\models}{\langle\alpha\rangle}A$.
Since $(W,{\leq},{R}){\models}{\langle\alpha\rangle}A{\rightarrow}B{\vee}{\lbrack\alpha\rbrack}(A{\rightarrow}C)$, then $(W,{\leq},{R},V){\models}{\langle\alpha\rangle}A{\rightarrow}B{\vee}{\lbrack\alpha\rbrack}(A{\rightarrow}C)$.
Consequently, $v{\models}{\langle\alpha\rangle}A{\rightarrow}B{\vee}{\lbrack\alpha\rbrack}(A{\rightarrow}C)$.
Since $v{\models}{\langle\alpha\rangle}A$, then $v{\models}B{\vee}{\lbrack\alpha\rbrack}(A{\rightarrow}C)$.
Hence, $v{\models}B$, or $v{\models}{\lbrack\alpha\rbrack}(A{\rightarrow}C)$.
In the former case, since $t{\geq}v$, then $t{\models}B$.
Thus, $t{\models}B{\vee}{\langle\alpha\rangle}C$: a contradiction.
In the latter case, since $v{R(\alpha)}u$, then $u{\models}A{\rightarrow}C$.
Since $u{\models}A$, then $u{\models}C$.
Since $t{\geq}{\circ}{R(\alpha)}u$, then $t{\models}{\langle\alpha\rangle}C$.
Consequently, $t{\models}B{\vee}{\langle\alpha\rangle}C$: a contradiction.
\medskip
\end{proof}
\begin{proposition}\label{proposition:validity:of:specific:axioms:doxastic:frames}
For all groups $\alpha$, the following formula is valid on any class of doxastic frames:\footnote{The mono-agent version of this formula~---~$p{\rightarrow}{\square}p$~---~has been considered in~\cite{Artemov:Protopopescu:2016}.
It can be read as follows: ``it is verified that $p$ has a proof once $p$ has been proved''.}
\begin{itemize}
\item $p{\rightarrow}{\lbrack\alpha\rbrack}p$.
\end{itemize}
\end{proposition}
\begin{proof}
If $p{\rightarrow}{\lbrack\alpha\rbrack}p$ is not valid on a class ${\mathcal C}$ of doxastic frames, there exists a formula $A$ and there exists a frame $(W,{\leq},{R})$ in ${\mathcal C}$ such that $(W,{\leq},{R}){\not\models}A{\rightarrow}{\lbrack\alpha\rbrack}A$.
Hence, there exists a model $(W,{\leq},{R},V)$ based on $(W,{\leq},{R})$ such that $(W,{\leq},{R},V){\not\models}A{\rightarrow}{\lbrack\alpha\rbrack}A$.
Thus, there exists $s{\in}W$ such that $s{\not\models}A{\rightarrow}{\lbrack\alpha\rbrack}A$.
Consequently, there exists $t{\in}W$ such that $s{\leq}t$, $t{\models}A$ and $t{\not\models}{\lbrack\alpha\rbrack}A$.
Hence, there exists $u{\in}W$ such that $t{\leq}{\circ}{R(\alpha)}u$ and $u{\not\models}A$.
Thus, there exists $v{\in}W$ such that $t{\leq}v$ and $v{R(\alpha)}u$.
Since $(W,{\leq},{R})$ is doxastic, then $v{\leq}u$.
Since $t{\leq}v$, then $t{\leq}u$.
Since $t{\models}A$, then $u{\models}A$: a contradiction.
\medskip
\end{proof}
\begin{proposition}\label{proposition:validity:of:specific:axioms:epistemic:frames}
For all groups $\alpha$, the following formula is valid on any class of epistemic frames:\footnote{Neither this formula, nor its mono-agent version~---~${\square}p{\rightarrow}{\neg}{\neg}{\lozenge}p$~---~have never been considered in the above-mentioned literature about intuitionistic modal logics.
The formula ${\lbrack\alpha\rbrack}p{\rightarrow}{\neg}{\neg}{\langle\alpha\rangle}p$ can be read as follows: ``it is not possible to produce a proof that the compatibility of $p$ with $\alpha$'s distributed knowledge cannot have a proof once it is verified by $\alpha$ that $p$ has a proof''.}
\begin{itemize}
\item ${\lbrack\alpha\rbrack}p{\rightarrow}{\neg}{\neg}{\langle\alpha\rangle}p$.
\end{itemize}
\end{proposition}
\begin{proof}
If ${\lbrack\alpha\rbrack}p{\rightarrow}{\neg}{\neg}{\langle\alpha\rangle}p$ is not valid on a class ${\mathcal C}$ of epistemic frames, there exists a formula $A$ and there exists a frame $(W,{\leq},{R})$ in ${\mathcal C}$ such that $(W,{\leq},{R}){\not\models}{\lbrack\alpha\rbrack}A{\rightarrow}{\neg}{\neg}{\langle\alpha\rangle}A$.
Hence, there exists a model $(W,{\leq},{R},V)$ based on $(W,{\leq},{R})$ such that $(W,{\leq},{R},V){\not\models}
$\linebreak$
{\lbrack\alpha\rbrack}A{\rightarrow}{\neg}{\neg}{\langle\alpha\rangle}A$.
Thus, there exists $s{\in}W$ such that $s{\not\models}{\lbrack\alpha\rbrack}A{\rightarrow}{\neg}{\neg}{\langle\alpha\rangle}A$.
Consequently, there exists $t{\in}W$ such that $s{\leq}t$, $t{\models}{\lbrack\alpha\rbrack}A$ and $t{\not\models}{\neg}{\neg}{\langle\alpha\rangle}A$.
Hence, there exists $u{\in}W$ such that $t{\leq}u$ and $u{\models}{\neg}{\langle\alpha\rangle}A$.
Since $(W,{\leq},{R})$ is epistemic, then there exists $v{\in}W$ such that $u{\leq}{\circ}{R(\alpha)}v$.
Thus, there exists $w{\in}W$ such that $u{\leq}w$ and $w{R(\alpha)}v$.
Since $u{\models}{\neg}{\langle\alpha\rangle}A$, then $w{\not\models}{\langle\alpha\rangle}A$.
Since $w{R(\alpha)}v$, then $v{\not\models}A$.
Since $u{\leq}{\circ}{R(\alpha)}v$, then $u{\not\models}{\lbrack\alpha\rbrack}A$.
Since $t{\leq}u$, then $t{\not\models}{\lbrack\alpha\rbrack}A$: a contradiction.
\medskip
\end{proof}
\begin{proposition}\label{proposition:validity:of:up:down}
For all groups $\alpha$, the following formulas are valid on any class of up and down reflexive and up and down symmetric frames:\footnote{Here, the reader should remind that as far as classical modal validity is concerned, the mono-agent versions of these formulas~---~${\square}p{\rightarrow}p$, $p{\rightarrow}{\lozenge}p$, $p{\rightarrow}{\square}{\lozenge}p$ and ${\lozenge}{\square}p{\rightarrow}p$~---~correspond to reflexivity and symmetry.
See~\cite[Chapter~$4$]{Blackburn:et:al:2001} and~\cite[Chapter~$3$]{Chagrov:Zakharyaschev:1997}.}
\begin{itemize}
\item ${\lbrack\alpha\rbrack}p{\rightarrow}p$,
\item $p{\rightarrow}{\langle\alpha\rangle}p$,
\item $p{\rightarrow}{\lbrack\alpha\rbrack}{\langle\alpha\rangle}p$,
\item ${\langle\alpha\rangle}{\lbrack\alpha\rbrack}p{\rightarrow}p$.
\end{itemize}
\end{proposition}
\begin{proof}
We only consider the $3$rd formula.
If $p{\rightarrow}{\lbrack\alpha\rbrack}{\langle\alpha\rangle}p$ is not valid in a class ${\mathcal C}$ of up and down reflexive and up and down symmetric frames, there exists a formula $A$ and there exists a frame $(W,{\leq},{R})$ in ${\mathcal C}$ such that $(W,{\leq},{R}){\not\models}A{\rightarrow}{\lbrack\alpha\rbrack}{\langle\alpha\rangle}A$.
Hence, there exists a model $(W,{\leq},{R},V)$ based on $(W,{\leq},{R})$ such that $(W,{\leq},{R},V){\not\models}A{\rightarrow}{\lbrack\alpha\rbrack}{\langle\alpha\rangle}A$.
Thus, there exists $s{\in}W$ such that $s{\not\models}A{\rightarrow}{\lbrack\alpha\rbrack}{\langle\alpha\rangle}A$.
Consequently, there exists $t{\in}W$ such that $s{\leq}t$, $t{\models}A$ and $t{\not\models}{\lbrack\alpha\rbrack}{\langle\alpha\rangle}A$.
Hence, there exists $u{\in}W$ such that $t{\leq}{\circ}{R(\alpha)}u$ and $u{\not\models}{\langle\alpha\rangle}A$.
Thus, there exists $v{\in}W$ such that $t{\leq}v$ and $v{R(\alpha)}u$.
Since $t{\models}A$, then $v{\models}A$.
Moreover, since $(W,{\leq},{R})$ is up and down symmetric, then $u{{\geq}{\circ}{R(\alpha)}{\circ}{\geq}}v$.
Consequently, there exists $w{\in}W$ such that $u{{\geq}{\circ}{R(\alpha)}}w$ and $w{\geq}v$.
Since $v{\models}A$, then $w{\models}A$.
Since $u{{\geq}{\circ}{R(\alpha)}}w$, then $u{\models}{\langle\alpha\rangle}A$: a contradiction.
\medskip
\end{proof}
\begin{proposition}\label{proposition:validity:of:axioms:inclusion}
For all groups $\alpha,\beta$, the following formulas are valid on any class of prestandard frames:\footnote{Within the classical context, these formulas are well-known to those who interest in epistemic logics with distributed knowledge~\cite{Halpern:Moses:1992,Wang:Agotnes:2020}.
Within the intuitionistic context, the $1$st formula has been considered in~\cite{Murai:Sano:2022,Su:et:al:2021} whereas the $2$nd formula has never been considered.}
\begin{itemize}
\item ${\lbrack\alpha\rbrack}p{\vee}{\lbrack\beta\rbrack}p{\rightarrow}{\lbrack\alpha{\cup}\beta\rbrack}p$,
\item ${\langle\alpha{\cup}\beta\rangle}p{\rightarrow}{\langle\alpha\rangle}p{\wedge}{\langle\beta\rangle}p$.
\end{itemize}
\end{proposition}
\begin{proof}
We only consider the $1$st formula.
If ${\lbrack\alpha\rbrack}p{\vee}{\lbrack\beta\rbrack}p{\rightarrow}{\lbrack\alpha{\cup}\beta\rbrack}p$ is not valid in a class ${\mathcal C}$ of prestandard frames, there exists a formula $A$ and there exists a frame $(W,{\leq},{R})$ in ${\mathcal C}$ such that $(W,{\leq},{R}){\not\models}{\lbrack\alpha\rbrack}A{\vee}{\lbrack\beta\rbrack}A{\rightarrow}{\lbrack\alpha{\cup}\beta\rbrack}A$.
Hence, there exists a model $(W,{\leq},{R},V)$ based on $(W,{\leq},{R})$ such that $(W,{\leq},{R},V){\not\models}{\lbrack\alpha\rbrack}A{\vee}{\lbrack\beta\rbrack}A{\rightarrow}{\lbrack\alpha{\cup}\beta\rbrack}A$.
Thus, there exists $s{\in}W$ such that $s{\not\models}{\lbrack\alpha\rbrack}A{\vee}{\lbrack\beta\rbrack}A{\rightarrow}{\lbrack\alpha{\cup}\beta\rbrack}A$.
Consequently, there exists $t{\in}W$ such that $s{\leq}t$, $t{\models}{\lbrack\alpha\rbrack}A{\vee}{\lbrack\beta\rbrack}A$ and $t{\not\models}{\lbrack\alpha{\cup}\beta\rbrack}A$.
Hence, there exists $u{\in}W$ such that $t{\leq}{\circ}{R(\alpha{\cup}\beta)}u$ and $u{\not\models}A$.
Thus, there exists $v{\in}W$ such that $t{\leq}v$ and $v{R(\alpha{\cup}\beta)}u$.
Since $(W,{\leq},{R})$ is prestandard, then $v{R(\alpha)}u$ and $v{R(\beta)}u$.
Since $t{\models}{\lbrack\alpha\rbrack}A{\vee}{\lbrack\beta\rbrack}A$, then $t{\models}{\lbrack\alpha\rbrack}A$, or $t{\models}{\lbrack\beta\rbrack}A$.
In the former case, since $t{\leq}v$ and $v{R(\alpha)}u$, then $u{\models}A$: a contradiction.
In the latter case, since $t{\leq}v$ and $v{R(\beta)}u$, then $u{\models}A$: a contradiction.
\medskip
\end{proof}
\section{Results about logics}
The following propositions will be used to show the completeness~---~with respect to their respective relational semantics~---~of the different intuitionistic modal logics that we will introduce in Section~\ref{section:axiomatization:completeness}.
In particular, by studying the proofs of Propositions~\ref{proposition:bounded:morphic:images:all}, \ref{proposition:bounded:morphic:images:partition}, \ref{proposition:all:tra:a}, \ref{proposition:refsym:updown:partition:a} and~\ref{proposition:refsym:updown:partition:b}, the reader will understand how difficult it is sometimes to show the completeness of such logics.\footnote{Here, we want to inform the reader that the proofs of Proposition~\ref{proposition:bounded:morphic:images:all}, \ref{proposition:bounded:morphic:images:partition}, \ref{proposition:all:tra:a}, \ref{proposition:refsym:updown:partition:a} and~\ref{proposition:refsym:updown:partition:b} have never been presented before.}
\begin{proposition}\label{proposition:bounded:morphic:images:all}
$\Log({\mathcal C}_{\all}^{\prestandard}){=}\Log({\mathcal C}_{\all}^{\standard})$.
\end{proposition}
\begin{proof}
Since ${\mathcal C}_{\all}^{\prestandard}$ contains ${\mathcal C}_{\all}^{\standard}$, then it suffices to prove that $\Log({\mathcal C}_{\all}^{\prestandard}){\supseteq}\Log({\mathcal C}_{\all}^{\standard})$.
If $\Log({\mathcal C}_{\all}^{\prestandard}){\not\supseteq}\Log({\mathcal C}_{\all}^{\standard})$, there exists a formula $A$ such that ${\mathcal C}_{\all}^{\prestandard}{\not\models}A$ and ${\mathcal C}_{\all}^{\standard}{\models}A$.
Hence, there exists a prestandard frame $(W,{\leq},{R})$ such that $(W,{\leq},{R}){\not\models}A$.
Thus, there exists a model $(W,{\leq},{R},V)$ based on $(W,{\leq},{R})$ such that $(W,{\leq},{R},V){\not\models}A$.
Consequently, there exists $s{\in}W$ such that $s{\not\models}A$.
For all groups $\alpha$, let $\pi(\alpha)\ :\ W{\times}W{\longrightarrow}
$\linebreak$
{\wp}(W)$ be the function such that for all $t,u{\in}W$, if $t{R(\alpha)}u$ then $\pi(\alpha)(t,u){=}\emptyset$ else $\pi(\alpha)(t,u){=}W$.
Let $I$ be the nonempty set of all functions from ${\wp^{\star}}(\Ag){\times}\Ag$ to ${\wp}(W)$.
Let $W^{\prime}{=}W{\times}I$.
Let $\leq^{\prime}$ be the preorder on $W^{\prime}$ such that for all $(t,g),(u,h){\in}W^{\prime}$, $(t,g){\leq^{\prime}}(u,h)$ if and only if $t{\leq}u$.
Let ${R^{\prime}}\ :\ {\wp^{\star}}(\Ag){\longrightarrow}\wp(W^{\prime}{\times}W^{\prime})$ be the function such that for all groups $\alpha$ and for all $(t,g),(u,h){\in}W^{\prime}$, $(t,g){R^{\prime}(\alpha)}(u,h)$ if and only if for all agents $a$ and for all groups $\gamma$,
\begin{itemize}
\item if $a{\in}\alpha$ and $a{\in}\gamma$ then $g(\gamma,a){\oplus}h(\gamma,a){=}\emptyset$,
\item $\Sigma\{g(\gamma,a):\ a{\in}\gamma\}{\oplus}\Sigma\{h(\gamma,a):\ a{\in}\gamma\}{=}\pi(\gamma)(t,u)$,
\end{itemize}
where $\oplus$ is the operation of symmetric difference in ${\wp}(W)$, $\Sigma\{g(\gamma,a):\ a{\in}\gamma\}$ denotes the result of applying $\oplus$ to the elements of $\{g(\gamma,a):\ a{\in}\gamma\}$ and $\Sigma\{h(\gamma,a):\ a{\in}\gamma\}$ denotes the result of applying $\oplus$ to the elements of $\{h(\gamma,a):\ a{\in}\gamma\}$.\footnote{As is well-known, the algebraic structure $({\wp}(W),{\emptyset},{W},{\oplus},{\cap})$ is a Boolean ring.
In particular, for all $X,Y{\in}{\wp}(W)$, $X{\oplus}Y{=}\emptyset$ if and only if $X{=}Y$.
In other respect, notice that since $\Ag$ is finite, then $\{g(\gamma,a):\ a{\in}\gamma\}$ and $\{h(\gamma,a):\ a{\in}\gamma\}$ are finite subsets of ${\wp}(W)$.}
Direct calculations would lead to the conclusion that for all groups $\alpha,\beta$, $R^{\prime}(\alpha{\cup}\beta){\subseteq}R^{\prime}(\alpha){\cap}R^{\prime}(\beta)$.
Hence, the frame $(W^{\prime},{\leq^{\prime}},{R^{\prime}})$ is prestandard.
\begin{claim}
The frame $(W^{\prime},{\leq^{\prime}},{R^{\prime}})$ is standard.
\end{claim}
\begin{proofclaim}
If the frame $(W^{\prime},{\leq^{\prime}},{R^{\prime}})$ is not standard, there exists groups $\alpha,\beta$ such that $R^{\prime}(\alpha{\cup}\beta){\not\supseteq}R^{\prime}(\alpha){\cap}R^{\prime}(\beta)$.
Hence, there exists $(t,g),(u,h){\in}W^{\prime}$ such that not $(t,g){R^{\prime}(\alpha}
$\linebreak$
{{\cup}\beta)}(u,h)$, $(t,g){R^{\prime}(\alpha)}(u,h)$ and $(t,g){R^{\prime}(\beta)}(u,h)$.
Thus, for all groups $\gamma$, $\Sigma\{g(\gamma,a):\ a{\in}\gamma\}{\oplus}\Sigma\{h(\gamma,a):\ a{\in}\gamma\}{=}\pi(\gamma)(t,u)$.
Since not $(t,g){R^{\prime}(\alpha{\cup}\beta)}(u,h)$, then there exists an agent $a$ and a group $\gamma$ such that $a{\in}\alpha{\cup}\beta$, $a{\in}\gamma$ and $g(\gamma,a){\oplus}h(\gamma,a){\not=}\emptyset$.
Consequently, $a{\in}\alpha$, or $a{\in}\beta$.
In the first case, since $a{\in}\gamma$ and $g(\gamma,a){\oplus}h(\gamma,a){\not=}\emptyset$, then not $(t,g){R^{\prime}(\alpha)}(u,h)$: a contradiction.
In the second case, since $a{\in}\gamma$ and $g(\gamma,a){\oplus}h(\gamma,a){\not=}
$\linebreak$
\emptyset$, then not $(t,g){R^{\prime}(\beta)}(u,h)$: a contradiction.
\medskip
\end{proofclaim}
\begin{claim}
For all groups $\alpha$, for all $t,u{\in}W$ and for all $g{\in}I$, the following conditions are equivalent:
\begin{enumerate}
\item $t{R(\alpha)}u$,
\item there exists $h{\in}I$ such that $(t,g){R^{\prime}(\alpha)}(u,h)$.
\end{enumerate}
\end{claim}
\begin{proofclaim}
$\mathbf{(1){\Rightarrow}(2)}:$
Suppose $t{R(\alpha)}u$.
Let $\tau\ :\ {\wp^{\star}}(\Ag){\longrightarrow}\Ag$ be a function such that for all groups $\gamma$, if $\gamma{\not\subseteq}\alpha$ then $\tau(\gamma){\in}\gamma$ and $\tau(\gamma){\not\in}\alpha$.\footnote{We can define such function by considering a total order on $\Ag$ and by saying that for all groups $\gamma$, $\tau(\gamma)$ is the greatest lower bound of $\gamma{\setminus}\alpha$.}
Let $h{\in}I$ be such that for all groups $\beta$ and for all agents $a$,
\begin{itemize}
\item if $a{\not\in}\beta$ then $h(\beta,a){=}\emptyset$,
\item if $a{\not\in}\alpha$, $a{\in}\beta$ and $a{\not=}\tau(\beta)$ then $h(\beta,a){=}\emptyset$,
\item if $a{\not\in}\alpha$, $a{\in}\beta$ and $a{=}\tau(\beta)$ then $h(\beta,a){=}\Sigma\{g(\beta,b):\ b{\in}\beta{\setminus}\alpha\}{\oplus}\pi(\gamma)(t,u)$,
\item if $a{\in}\alpha$ and $a{\in}\beta$ then $h(\beta,a){=}g(\beta,a)$.
\end{itemize}
Direct calculations would lead to the conclusion that $(t,g){R^{\prime}(\alpha)}(u,h)$.
\\
$\mathbf{(2){\Rightarrow}(1)}:$
Suppose there exists $h{\in}I$ such that $(t,g){R^{\prime}(\alpha)}(u,h)$.
Hence, for all agents $a$, if $a{\in}\alpha$ then $g(\alpha,a){\oplus}h(\alpha,a){=}\emptyset$.
Moreover, $\Sigma\{g(\alpha,a):\ a{\in}\alpha\}{\oplus}\Sigma\{h(\alpha,a):\ a{\in}\alpha\}
$\linebreak$
{=}\pi(\alpha)(t,u)$.
Thus, $\pi(\alpha)(t,u){=}\emptyset$.
Consequently, $t{R(\alpha)}u$.
\medskip
\end{proofclaim}
\\
Let $V^{\prime}\ :\ \At{\longrightarrow}\wp(W^{\prime})$ be the valuation on $(W^{\prime},{\leq^{\prime}},{R^{\prime}})$ such that for all atoms $p$, $V^{\prime}{p}{=}V(p){\times}I$.
\begin{claim}
For all formulas $B$ and for all $t{\in}W$, the following conditions are equivalent:
\begin{enumerate}
\item $t{\models}B$,
\item for all $g{\in}I$, $(t,g){\models}B$,
\item there exists $g{\in}I$ such that $(t,g){\models}B$.
\end{enumerate}
\end{claim}
\begin{proofclaim}
By induction on $B$.
\medskip
\end{proofclaim}
\\
Since $s{\not\models}A$, then there exists $f{\in}I$ such that $(s,f){\not\models}A$.
Thus, $(W^{\prime},{\leq^{\prime}},{R^{\prime}}){\not\models}A$.
Since $(W^{\prime},{\leq^{\prime}},{R^{\prime}})$ is standard, then ${\mathcal C}_{\all}^{\standard}{\not\models}A$: a contradiction.
\medskip
\end{proof}
\begin{proposition}\label{proposition:bounded:morphic:images:doxastic}
$\Log({\mathcal C}_{\doxastic}^{\prestandard}){=}\Log({\mathcal C}_{\doxastic}^{\standard})$.
\end{proposition}
\begin{proof}
Since ${\mathcal C}_{\doxastic}^{\prestandard}$ contains ${\mathcal C}_{\doxastic}^{\standard}$, then it suffices to prove that $\Log({\mathcal C}_{\doxastic}^{\prestandard}){\supseteq}\Log({\mathcal C}_{\doxastic}^{\standard})$.
If $\Log({\mathcal C}_{\doxastic}^{\prestandard}){\not\supseteq}\Log({\mathcal C}_{\doxastic}^{\standard})$, there exists a formula $A$ such that ${\mathcal C}_{\doxastic}^{\prestandard}{\not\models}A$ and ${\mathcal C}_{\doxastic}^{\standard}{\models}A$.
Hence, there exists a prestandard frame $(W,{\leq},{R})$ in ${\mathcal C}_{\doxastic}$ such that $(W,{\leq},{R}){\not\models}A$.
Thus, there exists a model $(W,{\leq},{R},V)$ based on $(W,{\leq},{R})$ such that $(W,{\leq},{R},V){\not\models}
$\linebreak$
A$.
Consequently, there exists $s{\in}W$ such that $s{\not\models}A$.
The rest of the proof is done by imitating the argument developed in the proof of Proposition~\ref{proposition:bounded:morphic:images:all}.\footnote{The reader may easily verify that since the prestandard frame $(W,{\leq},{R})$ is doxastic, then the frame $(W^{\prime},{\leq^{\prime}},{R^{\prime}})$ constructed in this argument is doxastic.}
\medskip
\end{proof}
\begin{proposition}\label{proposition:bounded:morphic:images:epistemic}
$\Log({\mathcal C}_{\epistemic}^{\prestandard}){=}\Log({\mathcal C}_{\epistemic}^{\standard})$.
\end{proposition}
\begin{proof}
Since ${\mathcal C}_{\epistemic}^{\prestandard}$ contains ${\mathcal C}_{\epistemic}^{\standard}$, then it suffices to prove that $\Log({\mathcal C}_{\epistemic}^{\prestandard}){\supseteq}\Log({\mathcal C}_{\epistemic}^{\standard})$.
If $\Log({\mathcal C}_{\epistemic}^{\prestandard}){\not\supseteq}\Log({\mathcal C}_{\epistemic}^{\standard})$, there exists a formula $A$ such that ${\mathcal C}_{\epistemic}{\not\models}A$ and ${\mathcal C}_{\epistemic}^{\standard}{\models}A$.
Thus, there exists a prestandard frame $(W,{\leq},{R})$ in ${\mathcal C}_{\epistemic}$ such that $(W,{\leq},{R}){\not\models}A$.
Consequently, there exists a model $(W,{\leq},{R},V)$ based on $(W,{\leq},{R})$ such that $(W,{\leq},{R},V)
$\linebreak$
{\not\models}A$.
Hence, there exists $s{\in}W$ such that $s{\not\models}A$.
The rest of the proof is done by imitating the argument developed in the proof of Proposition~\ref{proposition:bounded:morphic:images:all}.\footnote{The reader may easily verify that since the prestandard frame $(W,{\leq},{R})$ is epistemic, then the frame $(W^{\prime},{\leq^{\prime}},{R^{\prime}})$ constructed in this argument is epistemic.}
\medskip
\end{proof}
\begin{proposition}\label{proposition:bounded:morphic:images:updown}
$\Log({\mathcal C}_{\updown}^{\prestandard}){=}\Log({\mathcal C}_{\updown}^{\standard})$.
\end{proposition}
\begin{proof}
Since ${\mathcal C}_{\updown}^{\prestandard}$ contains ${\mathcal C}_{\updown}^{\standard}$, then it suffices to prove that $\Log({\mathcal C}_{\updown}^{\prestandard}){\supseteq}\Log({\mathcal C}_{\updown}^{\standard})$.
If $\Log({\mathcal C}_{\updown}^{\prestandard}){\not\supseteq}\Log({\mathcal C}_{\updown}^{\standard})$, there exists a formula $A$ such that ${\mathcal C}_{\updown}^{\prestandard}{\not\models}A$ and ${\mathcal C}_{\updown}^{\standard}{\models}A$.
Hence, there exists a prestandard frame $(W,{\leq},{R})$ in ${\mathcal C}_{\updown}$ such that $(W,{\leq},{R}){\not\models}A$.
Thus, there exists a model $(W,{\leq},{R},V)$ based on $(W,{\leq},{R})$ such that $(W,{\leq},{R},V){\not\models}
$\linebreak$
A$.
Consequently, there exists $s{\in}W$ such that $s{\not\models}A$.
The rest of the proof is done by imitating the argument developed in the proof of Proposition~\ref{proposition:bounded:morphic:images:all}.\footnote{The reader may easily verify that since the prestandard frame $(W,{\leq},{R})$ is up and down reflexive and up and down symmetric, then the frame $(W^{\prime},{\leq^{\prime}},{R^{\prime}})$ constructed in this argument is up and down reflexive and up and down symmetric.}
\medskip
\end{proof}
\begin{proposition}\label{proposition:bounded:morphic:images:refsym}
$\Log({\mathcal C}_{\refsym}^{\prestandard}){=}\Log({\mathcal C}_{\refsym}^{\standard})$.
\end{proposition}
\begin{proof}
Since ${\mathcal C}_{\refsym}^{\prestandard}$ contains ${\mathcal C}_{\refsym}^{\standard}$, then it suffices to prove that $\Log({\mathcal C}_{\refsym}^{\prestandard}){\supseteq}\Log({\mathcal C}_{\refsym}^{\standard})$.
If $\Log({\mathcal C}_{\refsym}^{\prestandard}){\not\supseteq}\Log({\mathcal C}_{\refsym}^{\standard})$, there exists a formula $A$ such that ${\mathcal C}_{\refsym}^{\prestandard}{\not\models}A$ and ${\mathcal C}_{\refsym}^{\standard}{\models}A$.
Hence, there exists a prestandard frame $(W,{\leq},{R})$ in ${\mathcal C}_{\refsym}$ such that $(W,{\leq},{R}){\not\models}A$.
Thus, there exists a model $(W,{\leq},{R},V)$ based on $(W,{\leq},{R})$ such that $(W,{\leq},{R},V){\not\models}A$.
Consequently, there exists $s{\in}W$ such that $s{\not\models}A$.
The rest of the proof is done by imitating the argument developed in the proof of Proposition~\ref{proposition:bounded:morphic:images:all}.\footnote{The reader may easily verify that since the prestandard frame $(W,{\leq},{R})$ is reflexive and symmetric, then the frame $(W^{\prime},{\leq^{\prime}},{R^{\prime}})$ constructed in this argument is reflexive and symmetric.}
\medskip
\end{proof}
\begin{proposition}\label{proposition:bounded:morphic:images:partition}
$\Log({\mathcal C}_{\partition}^{\prestandard}){=}\Log({\mathcal C}_{\partition}^{\standard})$.
\end{proposition}
\begin{proof}
Since ${\mathcal C}_{\partition}^{\prestandard}$ contains ${\mathcal C}_{\partition}^{\standard}$, then it suffices to prove that $\Log({\mathcal C}_{\partition}^{\prestandard}){\supseteq}\Log({\mathcal C}_{\partition}^{\standard})$.
If $\Log({\mathcal C}_{\partition}^{\prestandard}){\not\supseteq}\Log({\mathcal C}_{\partition}^{\standard})$, there exists a formula $A$ such that ${\mathcal C}_{\partition}^{\prestandard}{\not\models}A$ and ${\mathcal C}_{\partition}^{\standard}{\models}A$.
Hence, there exists a prestandard frame $(W,{\leq},{R})$ in ${\mathcal C}_{\partition}$ such that $(W,{\leq},{R}){\not\models}A$.
Thus, there exists a model $(W,{\leq},{R},V)$ based on $(W,{\leq},{R})$ such that $(W,{\leq},{R},V){\not\models}
$\linebreak$
A$.
Consequently, there exists $s{\in}W$ such that $s{\not\models}A$.
The rest of the proof is done by imitating the argument developed in the proof of Proposition~\ref{proposition:bounded:morphic:images:all}, the only change being for all groups $\alpha$, about the function $\pi(\alpha)\ :\ W{\times}W{\longrightarrow}{\wp}(W)$ that should be now defined by saying that for all $t,u{\in}W$, $\pi(\alpha)(t,u){=}\lbrack t\rbrack_{R(\alpha)}{\oplus}\lbrack u\rbrack_{R(\alpha)}$ where for all $v{\in}W$, $\lbrack v\rbrack_{R(\alpha)}$ denotes the equivalence class of $v$ modulo $R(\alpha)$.\footnote{The reader may easily verify that since the prestandard frame $(W,{\leq},{R})$ is a partition, then the frame $(W^{\prime},{\leq^{\prime}},{R^{\prime}})$ constructed in this argument is a partition.}
\medskip
\end{proof}
Propositions~\ref{proposition:bounded:morphic:images:all}, \ref{proposition:bounded:morphic:images:doxastic}, \ref{proposition:bounded:morphic:images:epistemic}, \ref{proposition:bounded:morphic:images:updown}, \ref{proposition:bounded:morphic:images:refsym} and~\ref{proposition:bounded:morphic:images:partition} say that our syntax and semantics cannot distinguish ${\mathcal C}^{\prestandard}$ from ${\mathcal C}^{\standard}$, ${\mathcal C}_{\doxastic}^{\prestandard}$ from ${\mathcal C}_{\doxastic}^{\standard}$, ${\mathcal C}_{\epistemic}^{\prestandard}$ from ${\mathcal C}_{\epistemic}^{\standard}$, ${\mathcal C}_{\updown}^{\prestandard}$ from ${\mathcal C}_{\updown}^{\standard}$, ${\mathcal C}_{\refsym}^{\prestandard}$ from ${\mathcal C}_{\refsym}^{\standard}$ and ${\mathcal C}_{\partition}^{\prestandard}$ from ${\mathcal C}_{\partition}^{\standard}$.
\begin{proposition}\label{proposition:all:tra:a}
$\Log({\mathcal C}_{\all}){=}\Log({\mathcal C}_{\transitive})$.
\end{proposition}
\begin{proof}
Since ${\mathcal C}_{\all}$ contains ${\mathcal C}_{\transitive}$, then it suffices to prove that $\Log({\mathcal C}_{\all}){\supseteq}\Log({\mathcal C}_{\transitive})$.
If $\Log({\mathcal C}_{\all}){\not\supseteq}\Log({\mathcal C}_{\transitive})$, there exists a formula $A$ such that ${\mathcal C}_{\all}{\not\models}A$ and ${\mathcal C}_{\transitive}{\models}A$.
Hence, there exists a frame $(W,{\leq},{R})$ such that $(W,{\leq},{R}){\not\models}A$.
Thus, there exists a model $(W,{\leq},{R},V)$ based on $(W,{\leq},{R})$ such that $(W,{\leq},{R},V){\not\models}A$.
Consequently, there exists $s{\in}W$ such that $s{\not\models}A$.
Let $W^{\prime}{=}W{\times}\{0,1\}$.
Let $\leq^{\prime}$ be the preorder on $W^{\prime}$ such that for all $(t,j),(u,k){\in}W^{\prime}$, $(t,j){\leq^{\prime}}(u,k)$ if and only if $t{\leq}u$.
Let ${R^{\prime}}\ :\ {\wp^{\star}}(\Ag){\longrightarrow}\wp(W^{\prime}
$\linebreak$
{\times}W^{\prime})$ be the function such that for all groups $\alpha$ and for all $(t,j),(u,k){\in}W^{\prime}$, $(t,j)
$\linebreak$
{R^{\prime}(\alpha)}(u,k)$ if and only if $t{R^{\prime}(\alpha)}u$, $j{\not=}1$ and $k{=}1$.
%
%
\begin{claim}
The frame $(W^{\prime},{\leq^{\prime}},{R^{\prime}})$ is transitive.
\end{claim}
\begin{proofclaim}
If the frame $(W^{\prime},{\leq^{\prime}},{R^{\prime}})$ is not transitive, there exists a group $\alpha$ and there exists $(t,j),(u,k),(v,l){\in}W^{\prime}$ such that $(t,j){R^{\prime}(\alpha)}(u,k)$ and $(u,k){R^{\prime}(\alpha)}(v,l)$.
Hence, $k{=}1$ and $k{\not=}1$: a contradiction.
\medskip
\end{proofclaim}
\begin{claim}
For all groups $\alpha$, for all $t,u,v{\in}W$ and for all $j{\in}\{0,1\}$, the following conditions are equivalent:
\begin{enumerate}
\item $t{\leq}u{R(\alpha)}v$,
\item there exists $k,l{\in}\{0,1\}$ such that $(t,j){\leq^{\prime}}(u,k){R^{\prime}(\alpha)}(v,l)$.
\end{enumerate}
\end{claim}
\begin{proofclaim}
Left to the reader.
\medskip
\end{proofclaim}
\begin{claim}
For all groups $\alpha$, for all $t,u,v{\in}W$ and for all $j{\in}\{0,1\}$, the following conditions are equivalent:
\begin{enumerate}
\item $t{\geq}u{R(\alpha)}v$,
\item there exists $k,l{\in}\{0,1\}$ such that $(t,j){\geq^{\prime}}(u,k){R^{\prime}(\alpha)}(v,l)$.
\end{enumerate}
\end{claim}
\begin{proofclaim}
Left to the reader.
\medskip
\end{proofclaim}
\\
Let $V^{\prime}\ :\ \At{\longrightarrow}\wp(W^{\prime})$ be the valuation on $(W^{\prime},{\leq^{\prime}},{R^{\prime}})$ such that for all atoms $p$, $V^{\prime}(p){=}V(p){\times}\{0,1\}$.
\begin{claim}
For all formulas $B$ and for all $t{\in}W$, the following conditions are equivalent:
\begin{enumerate}
\item $(W,{\leq},{R},V),t{\models}B$,
\item for all $j{\in}\{0,1\}$, $(W^{\prime},{\leq^{\prime}},{R^{\prime}},V^{\prime}),(t,j){\models}B$,
\item there exists $j{\in}\{0,1\}$ such that $(W^{\prime},{\leq^{\prime}},{R^{\prime}},V^{\prime}),(t,j){\models}B$.
\end{enumerate}
\end{claim}
\begin{proofclaim}
By induction on $B$.
\medskip
\end{proofclaim}
\\
Since $(W,{\leq},{R},V),s{\not\models}A$, then $(W^{\prime},{\leq^{\prime}},{R^{\prime}},V^{\prime}),(s,0){\not\models}A$.
Thus, $(W^{\prime},{\leq^{\prime}},{R^{\prime}},V^{\prime}){\not\models}
$\linebreak$
A$.
Consequently, $(W^{\prime},{\leq^{\prime}},{R^{\prime}}){\not\models}A$.
Since $(W^{\prime},{\leq^{\prime}},{R^{\prime}})$ is transitive, then ${\mathcal C}_{\transitive}{\not\models}A$: a contradiction.
\medskip
\end{proof}
\begin{proposition}\label{proposition:all:tra:b}
$\Log({\mathcal C}_{\all}^{\prestandard}){=}\Log({\mathcal C}_{\transitive}^{\prestandard})$.
\end{proposition}
\begin{proof}
Since ${\mathcal C}_{\all}^{\prestandard}$ contains ${\mathcal C}_{\transitive}^{\prestandard}$, then it suffices to prove that $\Log({\mathcal C}_{\all}^{\prestandard}){\supseteq}\Log({\mathcal C}_{\transitive}^{\prestandard})$.
If $\Log({\mathcal C}_{\all}^{\prestandard}){\not\supseteq}\Log({\mathcal C}_{\transitive}^{\prestandard})$, there exists a formula $A$ such that ${\mathcal C}_{\all}^{\prestandard}{\not\models}A$ and ${\mathcal C}_{\transitive}^{\prestandard}{\models}A$.
Hence, there exists a prestandard frame $(W,{\leq},{R})$ such that $(W,{\leq},{R}){\not\models}A$.
Thus, there exists a model $(W,{\leq},{R},V)$ based on $(W,{\leq},{R})$ such that $(W,{\leq},{R},V){\not\models}A$.
Consequently, there exists $s{\in}W$ such that $s{\not\models}A$.
The rest of the proof is done by imitating the argument developed in the proof of Proposition~\ref{proposition:all:tra:a}.\footnote{The reader may easily verify that since the frame $(W,{\leq},{R})$ is prestandard, then the frame $(W^{\prime},{\leq^{\prime}},{R^{\prime}})$ constructed in this argument is prestandard.}
\medskip
\end{proof}
\begin{proposition}\label{proposition:all:tra:c}
$\Log({\mathcal C}_{\all}^{\standard}){=}\Log({\mathcal C}_{\transitive}^{\standard})$.
\end{proposition}
\begin{proof}
Since ${\mathcal C}_{\all}^{\standard}$ contains ${\mathcal C}_{\transitive}^{\standard}$, then it suffices to prove that $\Log({\mathcal C}_{\all}^{\standard}){\supseteq}\Log({\mathcal C}_{\transitive}^{\standard})$.
If $\Log({\mathcal C}_{\all}^{\standard}){\not\supseteq}\Log({\mathcal C}_{\transitive}^{\standard})$, there exists a formula $A$ such that ${\mathcal C}_{\all}^{\standard}{\not\models}A$ and ${\mathcal C}_{\transitive}^{\standard}{\models}A$.
Hence, there exists a standard frame $(W,{\leq},{R})$ such that $(W,{\leq},{R}){\not\models}A$.
Thus, there exists a model $(W,{\leq},{R},V)$ based on $(W,{\leq},{R})$ such that $(W,{\leq},{R},V){\not\models}A$.
Consequently, there exists $s{\in}W$ such that $s{\not\models}A$.
The rest of the proof is done by imitating the argument developed in the proof of Proposition~\ref{proposition:all:tra:a}.\footnote{The reader may easily verify that since the frame $(W,{\leq},{R})$ is standard, then the frame $(W^{\prime},{\leq^{\prime}},{R^{\prime}})$ constructed in this argument is standard.}
\medskip
\end{proof}
Propositions~\ref{proposition:all:tra:a}, \ref{proposition:all:tra:b} and~\ref{proposition:all:tra:c} say that our syntax and semantics cannot distinguish ${\mathcal C}_{\all}$ from ${\mathcal C}_{\transitive}$, ${\mathcal C}_{\all}^{\prestandard}$ from ${\mathcal C}_{\transitive}^{\prestandard}$ and ${\mathcal C}_{\all}^{\standard}$ from ${\mathcal C}_{\transitive}^{\standard}$.
\begin{proposition}\label{proposition:refsym:updown:partition:a}
$\Log({\mathcal C}_{\updown}){=}\Log({\mathcal C}_{\refsym}){=}\Log({\mathcal C}_{\partition})$.
\end{proposition}
\begin{proof}
Since ${\mathcal C}_{\updown}$ contains ${\mathcal C}_{\refsym}$ and ${\mathcal C}_{\refsym}$ contains ${\mathcal C}_{\partition}$, then it suffices to prove that $\Log({\mathcal C}_{\updown}){\supseteq}\Log({\mathcal C}_{\partition})$.
If $\Log({\mathcal C}_{\updown}){\not\supseteq}\Log({\mathcal C}_{\partition})$, there exists a formula $A$ such that ${\mathcal C}_{\updown}{\not\models}A$ and ${\mathcal C}_{\partition}{\models}A$.
Hence, there exists a frame $(W,{\leq},{R})$ in ${\mathcal C}_{\updown}$ such that $(W,{\leq},{R})
$\linebreak$
{\not\models}A$.
Thus, there exists a model $(W,{\leq},{R},V)$ based on $(W,{\leq},{R})$ such that $(W,{\leq},{R},
$\linebreak$
V){\not\models}A$.
Consequently, there exists $s{\in}W$ such that $s{\not\models}A$.
Let $(W^{\prime},{\leq^{\prime}},{R^{\prime}})$ be the frame defined by $W^{\prime}{=}W$, ${\leq^{\prime}}{=}{\leq}$ and for all groups $\alpha$, ${R^{\prime}(\alpha)}{=}{({\leq}{\circ}{R(\alpha)}{\circ}{\leq}){\cap}({\geq}{\circ}{R(\alpha)}{\circ}{\geq})}$.
Since $(W,{\leq},{R})$ is up and down reflexive and up and down symmetric, then obviously, $(W^{\prime},{\leq^{\prime}},{R^{\prime}})$ is reflexive and symmetric.
Let $V^{\prime}\ :\ \At{\longrightarrow}\wp(W^{\prime})$ be the valuation on $(W^{\prime},{\leq^{\prime}},{R^{\prime}})$ such that for all atoms $p$, $V^{\prime}(p){=}V(p)$.
\begin{claim}
For all formulas $B$ and for all $t{\in}W$, the following conditions are equivalent:
\begin{enumerate}
\item $t{\models}B$,
\item $(W^{\prime},{\leq^{\prime}},{R^{\prime}},V^{\prime}),t{\models}B$.
\end{enumerate}
\end{claim}
\begin{proofclaim}
By induction on $B$.
\medskip
\end{proofclaim}
\\
Since $s{\not\models}A$, then $(W^{\prime},{\leq^{\prime}},{R^{\prime}},V^{\prime}),s{\not\models}A$.
Let $J$ be the nonempty set of all functions $g:\ W^{\prime}{\times}{\wp^{\star}}(\Ag){\longrightarrow}W^{\prime}$ such that for all $t{\in}W^{\prime}$ and for all $\beta{\in}{\wp^{\star}}(\Ag)$, $t{R^{\prime}(\beta)}g(t,\beta)$.\footnote{In order to prove that this set is nonempty, it suffices to consider the function $f:\ W^{\prime}{\times}{\wp^{\star}}(\Ag){\longrightarrow}W^{\prime}$ such that for all $t{\in}W^{\prime}$ and for all $\beta{\in}{\wp^{\star}}(\Ag)$, $f(t,\beta){=}t$.}
Let $W^{\prime\prime}{=}W^{\prime}{\times}J$.
Let ${\leq^{\prime\prime}}$ be the preorder on $W^{\prime\prime}$ such that for all $(t,g),(u,h){\in}W^{\prime\prime}$, $(t,g){\leq^{\prime\prime}}(u,h)$ if and only if $t{\leq^{\prime}}u$.
Let ${R^{\prime\prime}}\ :\ {\wp^{\star}}(\Ag){\longrightarrow}\wp(W^{\prime\prime}{\times}W^{\prime\prime})$ be the function such that for all groups $\beta$ and for all $(t,g),(u,h){\in}W^{\prime\prime}$, $(t,g){R^{\prime\prime}(\beta)}(u,h)$ if and only if $t{R^{\prime}(\beta)}u$ and $\{t,g(t,\beta)\}{=}\{u,h(u,\beta)\}$.
\begin{claim}
The frame $(W^{\prime\prime},{\leq^{\prime\prime}},{R^{\prime\prime}})$ is a partition.
\end{claim}
\begin{proofclaim}
If the frame $(W^{\prime\prime},{\leq^{\prime\prime}},{R^{\prime\prime}})$ is not a partition, the frame $(W^{\prime\prime},{\leq^{\prime\prime}},{R^{\prime\prime}})$ is not reflexive, or the frame $(W^{\prime\prime},{\leq^{\prime\prime}},{R^{\prime\prime}})$ is not symmetric, or the frame $(W^{\prime\prime},{\leq^{\prime\prime}},{R^{\prime\prime}})$ is not transitive.
Leaving the $1$st and $2$nd cases to the reader, we only consider the $3$rd case.
Hence, there exists a group $\alpha$ and there exists $(t,g),(u,h),(v,i){\in}W^{\prime\prime}$ such that $(t,g){R^{\prime\prime}(\alpha)}(u,h)$, $(u,h){R^{\prime\prime}(\alpha)}(v,i)$ and not $(t,g){R^{\prime\prime}(\alpha)}(v,i)$.
Thus, $t{R^{\prime}(\alpha)}u$, $\{t,g(t,\alpha)\}{=}\{u,h(u,\alpha)\}$, $u{R^{\prime}(\alpha)}v$ and $\{u,h(u,\alpha)\}{=}\{v,i(v,\alpha)\}$.
Moreover, not $t
$\linebreak$
{R^{\prime}(\alpha)}v$, or $\{t,g(t,\alpha)\}{\not=}\{v,i(v,\alpha)\}$.
Since $\{t,g(t,\alpha)\}{=}\{u,h(u,\alpha)\}$ and $\{u,h(u,\alpha)\}
$\linebreak$
{=}\{v,i(v,\alpha)\}$, then $\{t,g(t,\alpha)\}{=}\{v,i(v,\alpha)\}$.
Since not $t{R^{\prime}(\alpha)}v$, or $\{t,g(t,\alpha)\}{\not=}\{v,
$\linebreak$
i(v,\alpha)\}$, then not $t{R^{\prime}(\alpha)}v$.
Since $\{t,g(t,\alpha)\}{=}\{u,h(u,\alpha)\}$, then $t{=}u$, or $t{=}h(u,\alpha)$.
In the former case, since $u{R^{\prime}(\alpha)}v$, then $t{R^{\prime}(\alpha)}v$: a contradiction.
In the latter case, since $\{u,h(u,\alpha)\}{=}\{v,i(v,\alpha)\}$, then $u{=}v$, or $h(u,\alpha){=}v$.
In the former case, since $(W^{\prime},{\leq^{\prime}},{R^{\prime}})$ is symmetric and $t{=}h(u,\alpha)$, then $t{R^{\prime}(\alpha)}v$: a contradiction.
In the latter case, since $(W^{\prime},{\leq^{\prime}},{R^{\prime}})$ is reflexive and $t{=}h(u,\alpha)$, then $t{R^{\prime}(\alpha)}v$: a contradiction.
\medskip
\end{proofclaim}
\begin{claim}
For all groups $\alpha$, for all $t,u,v{\in}W^{\prime}$ and for all $g{\in}J$, the following conditions are equivalent:
\begin{enumerate}
\item $t{\leq^{\prime}}u{R^{\prime}(\alpha)}v$,
\item there exists $h,i{\in}J$ such that $(t,g){\leq^{\prime\prime}}(u,h){R^{\prime\prime}(\alpha)}(v,i)$.
\end{enumerate}
\end{claim}
\begin{proofclaim}
$\mathbf{(1){\Rightarrow}(2)}:$
Suppose $t{\leq^{\prime}}u{R^{\prime}(\alpha)}v$.
Let $h,i{\in}J$ be such that for all $w{\in}W^{\prime}$ and for all groups $\beta$,
\begin{itemize}
\item if $w{\not=}u$, or $\beta{\not=}\alpha$ then $h(w,\beta){=}w$,
\item if $w{\not=}v$, or $\beta{\not=}\alpha$ then $i(w,\beta){=}w$,
\item if $w{=}u$ and $\beta{=}\alpha$ then $h(w,\beta){=}v$,
\item if $w{=}v$ and $\beta{=}\alpha$ then $i(w,\beta){=}u$.
\end{itemize}
Direct calculations would lead to the conclusion that $(t,g){\leq^{\prime\prime}}(u,h){R^{\prime\prime}(\alpha)}(v,i)$.
\\
$\mathbf{(2){\Rightarrow}(1)}:$
Suppose there exists $h,i{\in}J$ such that $(t,g){\leq^{\prime\prime}}(u,h){R^{\prime\prime}(\alpha)}(v,i)$.
Consequently, $t{\leq^{\prime}}u{R^{\prime}(\alpha)}v$.
\medskip
\end{proofclaim}
\begin{claim}
For all groups $\alpha$, for all $t,u,v{\in}W^{\prime}$ and for all $g{\in}J$, the following conditions are equivalent:
\begin{enumerate}
\item $t{\geq^{\prime}}u{R^{\prime}(\alpha)}v$,
\item there exists $h,i{\in}J$ such that $(t,g){\geq^{\prime\prime}}(u,h){R^{\prime\prime}(\alpha)}(v,i)$.
\end{enumerate}
\end{claim}
\begin{proofclaim}
Similar to the proof of the previous claim.
\medskip
\end{proofclaim}
\\
Let $V^{\prime\prime}\ :\ \At{\longrightarrow}\wp(W^{\prime\prime})$ be the valuation on $(W^{\prime\prime},{\leq^{\prime\prime}},{R^{\prime\prime}})$ such that for all atoms $p$, $V^{\prime\prime}(p){=}\{(t,g):\ (t,g){\in}W^{\prime\prime}$ is such that $t{\in}V^{\prime}(p)\}$.
\begin{claim}
For all formulas $B$ and for all $t{\in}W^{\prime}$, the following conditions are equivalent:
\begin{enumerate}
\item $(W^{\prime},{\leq^{\prime}},{R^{\prime}},V^{\prime}),t{\models}B$,
\item for all $g{\in}J$, $(W^{\prime\prime},{\leq^{\prime\prime}},{R^{\prime\prime}},V^{\prime\prime}),(t,g){\models}B$,
\item there exists $g{\in}J$ such that $(W^{\prime\prime},{\leq^{\prime\prime}},{R^{\prime\prime}},V^{\prime\prime}),(t,g){\models}B$.
\end{enumerate}
\end{claim}
\begin{proofclaim}
By induction on $B$.
\medskip
\end{proofclaim}
\\
Since $(W^{\prime},{\leq^{\prime}},{R^{\prime}},V^{\prime}),s{\not\models}A$, then $(W^{\prime\prime},{\leq^{\prime\prime}},{R^{\prime\prime}},V^{\prime\prime}),(s,f){\not\models}A$ where $f{\in}J$ is such that for all $t{\in}W^{\prime}$ and for all $\beta{\in}{\wp^{\star}}(\Ag)$, $f(t,\beta){=}t$.
Thus, $(W^{\prime\prime},{\leq^{\prime\prime}},{R^{\prime\prime}},V^{\prime\prime}){\not\models}A$.
Consequently, $(W^{\prime\prime},{\leq^{\prime\prime}},{R^{\prime\prime}}){\not\models}A$.
Since $(W^{\prime\prime},{\leq^{\prime\prime}},{R^{\prime\prime}})$ is a partition, then ${\mathcal C}_{\partition}{\not\models}A$: a contradiction.
\medskip
\end{proof}
\begin{proposition}\label{proposition:refsym:updown:partition:b}
$\Log({\mathcal C}_{\updown}^{\prestandard}){=}\Log({\mathcal C}_{\refsym}^{\prestandard}){=}\Log({\mathcal C}_{\partition}^{\prestandard})$.
\end{proposition}
\begin{proof}
Since ${\mathcal C}_{\updown}^{\prestandard}$ contains ${\mathcal C}_{\refsym}^{\prestandard}$ and ${\mathcal C}_{\refsym}^{\prestandard}$ contains ${\mathcal C}_{\partition}^{\prestandard}$, then it suffices to prove that $\Log({\mathcal C}_{\updown}^{\prestandard}){\supseteq}\Log({\mathcal C}_{\partition}^{\prestandard})$.
If $\Log({\mathcal C}_{\updown}^{\prestandard}){\not\supseteq}\Log({\mathcal C}_{\partition}^{\prestandard})$, there exists a formula $A$ such that ${\mathcal C}_{\updown}^{\prestandard}{\not\models}A$ and ${\mathcal C}_{\partition}^{\prestandard}{\models}A$.
Hence, there exists a prestandard frame $(W,{\leq},{R})$ in ${\mathcal C}_{\updown}$ such that $(W,{\leq},{R}){\not\models}A$.
Thus, there exists a model $(W,{\leq},{R},V)$ based on $(W,{\leq},{R})$ such that $(W,{\leq},{R},V){\not\models}A$.
Consequently, there exists $s{\in}W$ such that $s{\not\models}A$.
The rest of the proof is done by imitating the argument developed in the proof of Proposition~\ref{proposition:refsym:updown:partition:a}, the only change being for all groups $\beta$, about the binary relation $R^{\prime\prime}(\beta)$ that should be now defined by saying that for all $(t,g),(u,h){\in}W^{\prime\prime}$, $(t,g){R^{\prime\prime}(\beta)}(u,h)$ if and only if $t{R^{\prime}(\beta)}u$ and for all groups $\gamma$, if $\gamma{\subseteq}\beta$ then $\{t,g(t,\gamma)\}{=}\{u,h(u,\gamma)\}$.\footnote{The reader may easily verify that since the frame $(W,{\leq},{R})$ is prestandard, then the frames $(W^{\prime},{\leq^{\prime}},{R^{\prime}})$ and $(W^{\prime\prime},{\leq^{\prime\prime}},{R^{\prime\prime}})$ constructed in this argument are prestandard.}
\medskip
\end{proof}
\begin{proposition}\label{proposition:refsym:updown:partition:c}
$\Log({\mathcal C}_{\updown}^{\standard}){=}\Log({\mathcal C}_{\refsym}^{\standard}){=}\Log({\mathcal C}_{\partition}^{\standard})$.
\end{proposition}
\begin{proof}
By Propositions~\ref{proposition:bounded:morphic:images:updown}, \ref{proposition:bounded:morphic:images:refsym}, \ref{proposition:bounded:morphic:images:partition} and~\ref{proposition:refsym:updown:partition:b}.
\medskip
\end{proof}
Propositions~\ref{proposition:refsym:updown:partition:a}, \ref{proposition:refsym:updown:partition:b} and~\ref{proposition:refsym:updown:partition:c} say that our syntax and semantics cannot distinguish ${\mathcal C}_{\updown}$ from ${\mathcal C}_{\refsym}$ and from ${\mathcal C}_{\partition}$, ${\mathcal C}_{\updown}^{\prestandard}$ from ${\mathcal C}_{\refsym}^{\prestandard}$ and from ${\mathcal C}_{\partition}^{\prestandard}$ and ${\mathcal C}_{\updown}^{\standard}$ from ${\mathcal C}_{\refsym}^{\standard}$ and from ${\mathcal C}_{\partition}^{\standard}$.
\section{Axiomatization and canonical model construction}\label{section:axiomatization:completeness}
\paragraph{Intuitionistic modal logics}
An {\em intuitionistic modal logic}\/ is a set of formulas closed for uniform substitution, containing the standard axioms of $\IPL$ and axioms
%
%
\begin{description}
\item[$(\Axiom1)$] ${\lbrack\alpha\rbrack}p{\wedge}{\lbrack\alpha\rbrack}q{\rightarrow}{\lbrack\alpha\rbrack}(p{\wedge}q)$,
\item[$(\Axiom2)$] ${\langle\alpha\rangle}(p{\vee}q){\rightarrow}{\langle\alpha\rangle}p{\vee}{\langle\alpha\rangle}q$,
\item[$(\Axiom3)$] ${\lbrack\alpha\rbrack}{\top}$,
\item[$(\Axiom4)$] ${\neg}{\langle\alpha\rangle}{\bot}$,
\item[$(\Axiom5)$] ${\lbrack\alpha\rbrack}(p{\vee}q){\rightarrow}(({\langle\alpha\rangle}p{\rightarrow}{\lbrack\alpha\rbrack}q){\rightarrow}{\lbrack\alpha\rbrack}q)$,
\end{description}
and closed with respect to the standard rules of $\IPL$ and the rules
%
%
\begin{description}
\item[$(\Rule1)$] $\frac{p{\rightarrow}q}{{\lbrack\alpha\rbrack}p{\rightarrow}{\lbrack\alpha\rbrack}q}$,
\item[$(\Rule2)$] $\frac{p{\rightarrow}q}{{\langle\alpha\rangle}p{\rightarrow}{\langle\alpha\rangle}q}$,
\item[$(\Rule3)$] $\frac{{\langle\alpha\rangle}p{\rightarrow}q{\vee}{\lbrack\alpha\rbrack}(p{\rightarrow}r)}{{\langle\alpha\rangle}p{\rightarrow}q{\vee}{\langle\alpha\rangle}r}$.
\end{description}
Let $\L_{\all}$, $\L_{\doxastic}$, $\L_{\epistemic}$ and $\L_{\partition}$ be respectively the least intuitionistic modal logic, the least intuitionistic modal logic containing axiom
\begin{description}
\item[$(\Axiom6)$] $p{\rightarrow}{\lbrack\alpha\rbrack}p$,
\end{description}
the least intuitionistic modal logic containing axioms~$(\Axiom6)$ and
\begin{description}
\item[$(\Axiom7)$] ${\lbrack\alpha\rbrack}p{\rightarrow}{\neg}{\neg}{\langle\alpha\rangle}p$
\end{description}
and the least intuitionistic modal logic containing axioms
\begin{description}
\item[$(\Axiom8)$] ${\lbrack\alpha\rbrack}p{\rightarrow}p$,
\item[$(\Axiom9)$] $p{\rightarrow}{\langle\alpha\rangle}p$,
\item[$(\Axiom10)$] $p{\rightarrow}{\lbrack\alpha\rbrack}{\langle\alpha\rangle}p$,
\item[$(\Axiom11)$] ${\langle\alpha\rangle}{\lbrack\alpha\rbrack}p{\rightarrow}p$.
\end{description}
Let $\L_{\all}^{\Distrib}$, $\L_{\doxastic}^{\Distrib}$, $\L_{\epistemic}^{\Distrib}$ and $\L_{\partition}^{\Distrib}$ be respectively the least intuitionistic modal logic containing $\L_{\all}$ and axioms
\begin{description}
\item[$(\Axiom12)$] ${\lbrack\alpha\rbrack}p{\vee}{\lbrack\beta\rbrack}p{\rightarrow}{\lbrack\alpha{\cup}\beta\rbrack}p$,
\item[$(\Axiom13)$] ${\langle\alpha{\cup}\beta\rangle}p{\rightarrow}{\langle\alpha\rangle}p{\wedge}{\langle\beta\rangle}p$,
\end{description}
the least intuitionistic modal logic containing $\L_{\doxastic}$ and axioms~$(\Axiom12)$ and~$(\Axiom13)$, the least intuitionistic modal logic containing $\L_{\epistemic}$ and axioms~$(\Axiom12)$ and~$(\Axiom13)$ and the least intuitionistic modal logic containing $\L_{\partition}$ and axioms~$(\Axiom12)$ and~$(\Axiom13)$.
\paragraph{Theories}
Let $\L$ be an intuitionistic modal logic.\footnote{The tools and techniques considered in this paragraph will be used to show the completeness~---~with respect to their respective relational semantics~---~of the above-mentioned intuitionistic modal logics.}
A {\em $\L$-theory}\/ is a set of formulas containing $\L$ and closed under modus ponens.
Obviously, for all indexed families $(\Gamma_{i})_{i{\in}I}$ of $\L$-theories, $\bigcap\{\Gamma_{i}\ :\ i{\in}I\}$ is a $\L$-theory and for all nonempty chains $(\Gamma_{i})_{i{\in}I}$ of $\L$-theories, $\bigcup\{\Gamma_{i}\ :\ i{\in}I\}$ is a $\L$-theory.
As a result, there exists a least $\L$-theory (which is nothing but $\L$).
In other respect, obviously, $\Fo$ is the greatest $\L$-theory.
A $\L$-theory $\Gamma$ is {\em proper}\/ if ${\bot}{\not\in}\Gamma$.
Obviously, for all $\L$-theories $\Gamma$, $\Gamma$ is proper if and only if $\Gamma{\not=}\Fo$.
Moreover, for all proper $\L$-theories $\Gamma$ and for all groups $\alpha$, ${\langle\alpha\rangle}{\bot}{\not\in}\Gamma$.
A proper $\L$-theory $\Gamma$ is {\em prime}\/ if for all formulas $A,B$, if $A{\vee}B{\in}\Gamma$ then $A{\in}\Gamma$, or $B{\in}\Gamma$.
For all $\L$-theories $\Gamma$ and for all sets $\Delta$ of formulas, let $\Gamma{+}\Delta{=}\{B{\in}\Fo\ :$ there exists $m{\in}\N$ and there exists $A_{1},\ldots,A_{m}{\in}\Delta$ such that $A_{1}{\wedge}\ldots{\wedge}A_{m}{\rightarrow}B{\in}\Gamma\}$.
For all $\L$-theories $\Gamma$ and and for all formulas $A$, we write $\Gamma{+}A$ instead of $\Gamma{+}\{A\}$.
For all $\L$-theories $\Gamma$ and for all groups $\alpha$, let ${\lbrack\alpha\rbrack}\Gamma{=}\{A{\in}\Fo\ :\ {\lbrack\alpha\rbrack}A{\in}\Gamma\}$.
Using the standard axioms of $\IPL$ and the standard rules of $\IPL$, the reader may easily verify that for all $\L$-theories $\Gamma$ and for all sets $\Delta$ of formulas,
\begin{enumerate}
\item $\Gamma{+}\Delta$ is a $\L$-theory,
\item $\Gamma{\subseteq}\Gamma{+}\Delta$,
\item $\Delta{\subseteq}\Gamma{+}\Delta$,
\item for all $\L$-theories $\Lambda$, if $\Gamma{\subseteq}\Lambda$ and $\Delta{\subseteq}\Lambda$ then $\Gamma{+}\Delta{\subseteq}\Lambda$,
\item $\Gamma{+}\Delta$ is proper if and only if for all $m{\in}\N$ and for all $A_{1},\ldots,A_{m}{\in}\Delta$, ${\neg}(A_{1}{\wedge}\ldots{\wedge}
$\linebreak$
A_{m}){\not\in}\Gamma$.
\end{enumerate}
Moreover, using axioms $(\Axiom1)$ and $(\Axiom3)$ and rule $(\Rule1)$, the reader may easily verify that for all $\L$-theories $\Gamma$ and for all groups $\alpha$, ${\lbrack\alpha\rbrack}\Gamma$ is a $\L$-theory.
\paragraph{Existence properties}
Let $\L$ be an intuitionistic modal logic.
The following propositions will be used in the canonical model construction.
See~\cite[Lemmas~$21$, $22$, $23$ and~$24$]{Balbiani:Gencer:APAL:2024} for their proofs in the mono-agent case.
\begin{proposition}\label{lemma:prime:proper:for:implication}
Let $\Gamma$ be a prime $\L$-theory.
Let $B,C$ be formulas.
If $B{\rightarrow}C{\not\in}\Gamma$ then there exists a prime $\L$-theory $\Delta$ such that $\Gamma{\subseteq}\Delta$, $B{\in}\Delta$ and $C{\not\in}\Delta$.
\end{proposition}
\begin{proposition}\label{lemma:prime:proper:for:square}
Let $\Gamma$ be a prime $\L$-theory.
Let $B$ be a formula and $\beta$ be a group.
If ${\lbrack\beta\rbrack}B{\not\in}\Gamma$ then there exists prime $\L$-theories $\Delta,\Lambda$ such that $\Gamma{\subseteq}\Delta$, ${\lbrack\beta\rbrack}\Delta{\subseteq}\Lambda$, ${\langle\beta\rangle}\Lambda{\subseteq}\Delta$ and $B{\not\in}\Lambda$.
\end{proposition}
\begin{proposition}\label{lemma:prime:proper:for:lozenge}
Let $\Gamma$ be a prime $\L$-theory.
Let $B$ be a formula and $\beta$ be a group.
If ${\langle\beta\rangle}B{\in}\Gamma$ then there exists prime $\L$-theories $\Delta,\Lambda$ such that $\Gamma{{\supseteq}}\Delta$, ${\lbrack\beta\rbrack}\Delta{\subseteq}\Lambda$, ${\langle\beta\rangle}\Lambda{\subseteq}\Delta$ and $B{\in}\Lambda$.
\end{proposition}
\begin{proposition}\label{lemma:almost:completeness}
Let $A$ be a formula.
If $A{\not\in}\L$ then there exists a prime $\L$-theory $\Gamma$ such that $A{\not\in}\Gamma$.
\end{proposition}
\paragraph{Canonical model construction}
Let $\L$ be an intuitionistic modal logic.
The {\em canonical frame of $\L$}\/ is the relational structure $(W_{\L},{\leq_{\L}},{R_{\L}})$ where $W_{\L}$ is the nonempty set of all prime $\L$-theories, $\leq_{\L}$ is the preorder on $W_{\L}$ such that for all $\Gamma,\Delta{\in}W_{\L}$, $\Gamma{\leq_{\L}}\Delta$ if and only if $\Gamma{\subseteq}\Delta$ and ${R_{\L}}\ :\ {\wp^{\star}}(\Ag){\longrightarrow}\wp(W_{\L}{\times}W_{\L})$ is the function such that for all groups $\alpha$ and for all $\Gamma,\Delta{\in}W_{\L}$, $\Gamma{R_{\L}(\alpha)}\Delta$ if and only if ${\lbrack\alpha\rbrack}\Gamma{\subseteq}\Delta$ and ${\langle\alpha\rangle}\Delta{\subseteq}\Gamma$.
The following proposition will be used in the canonical model construction.
See~\cite[Lemmas~$30$, $31$, $32$ and~$33$]{Balbiani:Gencer:APAL:2024} for a proof of its Item~$3$ in the mono-agent case.
\begin{proposition}\label{proposition:A8:A9:canonical:frame}
\begin{enumerate}
\item If $\L$ contains $(\Axiom6)$ then $(W_{\L},{\leq_{\L}},{R_{\L}})$ is doxastic,
\item If $\L$ contains $(\Axiom6)$ and $(\Axiom7)$ then $(W_{\L},{\leq_{\L}},{R_{\L}})$ is epistemic,
\item If $\L$ contains $(\Axiom8)$, $(\Axiom9)$, $(\Axiom10)$ and $(\Axiom11)$ then $(W_{\L},{\leq_{\L}},{R_{\L}})$ is up and down reflexive and up and down symmetric.
\end{enumerate}
\end{proposition}
\begin{proof}
$\mathbf{(1)}$~Suppose $\L$ contains $(\Axiom6)$.
For the sake of the contradiction, suppose $(W_{\L},
$\linebreak$
{\leq_{\L}},{R_{\L}})$ is not doxastic.
Hence, there exists a group $\alpha$ and there exists $\Gamma,\Delta{\in}W_{\L}$ such that $\Gamma{R_{\L}(\alpha)}\Delta$ and not $\Gamma{\leq_{\L}}\Delta$.
Thus, $\Gamma{\not\subseteq}\Delta$.
Consequently, there exists a formula $A$ such that $A{\in}\Gamma$ and $A{\not\in}\Delta$.
Since $\L$ contains $(\Axiom6)$, then $A{\rightarrow}{\lbrack\alpha\rbrack}A{\in}\Gamma$.
Since $A{\in}\Gamma$, then ${\lbrack\alpha\rbrack}A{\in}\Gamma$.
Hence, $A{\in}{\lbrack\alpha\rbrack}\Gamma$.
Since $\Gamma{R_{\L}(\alpha)}\Delta$, then ${\lbrack\alpha\rbrack}\Gamma{\subseteq}\Delta$.
Since $A{\in}{\lbrack\alpha\rbrack}\Gamma$, then $A{\in}\Delta$: a contradiction.
\\
$\mathbf{(2)}$~Suppose $\L$ contains $(\Axiom6)$ and $(\Axiom7)$.
Thus, by Item~$\mathbf{(1)}$, $(W_{\L},{\leq_{\L}},{R_{\L}})$ is doxastic.
We demonstrate that $(W_{\L},{\leq_{\L}},{R_{\L}})$ is epistemic.
Let $\Gamma{\in}W_{\L}$ and $\alpha$ be a group.
Let ${\mathcal S}{=}\{\Delta\ :\ \Delta$ is a $\L$-theory such that {\bf (1)}~$\Gamma{\subseteq}\Delta$ and {\bf (2)}~${\bot}{\not\in}\Delta\}$.
Obviously, $\Gamma\in{\mathcal S}$.
Consequently, ${\mathcal S}$ is nonempty.
Moreover, for all nonempty chains $(\Delta_{i})_{i{\in}I}$ of elements of ${\mathcal S}$, $\bigcup\{\Delta_{i}\ :\ i{\in}I\}$ is an element of ${\mathcal S}$.
Hence, by Zorn's Lemma, ${\mathcal S}$ possesses a maximal element $\Delta$.
Thus, $\Delta$ is a $\L$-theory such that $\Gamma{\subseteq}\Delta$ and ${\bot}{\not\in}\Delta$.
We claim that $\Delta$ is proper.
If not, ${\bot}{\in}\Delta$: a contradiction.
Consequently, $\Delta$ is proper.
We claim that $\Delta$ is prime.
If not, there exists formulas $C,D$ such that $C{\vee}D{\in}\Delta$, $C{\not\in}\Delta$ and $D{\not\in}\Delta$.
Hence, by the maximality of $\Delta$ in ${\mathcal S}$, $\Delta{+}C{\not\in}{\mathcal S}$ and $\Delta{+}D{\not\in}{\mathcal S}$.
Thus, ${\bot}{\in}\Delta{+}C$ and ${\bot}{\in}\Delta{+}D$.
Consequently, $C{\rightarrow}{\bot}{\in}\Delta$ and $D{\rightarrow}{\bot}{\in}\Delta$.
Hence, $C{\vee}D{\rightarrow}{\bot}{\in}\Delta$.
Since $C{\vee}D{\in}\Delta$, then ${\bot}{\in}\Delta$: a contradiction.
Thus, $\Delta$ is prime.
Let ${\mathcal T}{=}\{\Lambda\ :\ \Lambda$ is a $\L$-theory such that ${\lbrack\alpha\rbrack}\Delta{\subseteq}\Lambda$ and ${\langle\alpha\rangle}\Lambda{\subseteq}\Delta\}$.
We claim that ${\lbrack\alpha\rbrack}\Delta{\in}{\mathcal T}$.
If not, there exists a formula $C$ such that $C{\in}{\lbrack\alpha\rbrack}\Delta$ and ${\langle\alpha\rangle}C{\not\in}\Delta$.
Thus, ${\lbrack\alpha\rbrack}C{\in}\Delta$.
Moreover, by the maximality of $\Delta$ in ${\mathcal S}$, $\Delta{+}{\langle\alpha\rangle}C{\not\in}{\mathcal S}$.
Consequently, ${\bot}{\in}\Delta{+}{\langle\alpha\rangle}C$.
Hence, ${\langle\alpha\rangle}C{\rightarrow}{\bot}{\in}\Delta$.
Since $\L$ contains $(\Axiom7)$ and ${\lbrack\alpha\rbrack}C{\in}\Delta$, then ${\neg}{\neg}{\langle\alpha\rangle}C{\in}\Delta$.
Since ${\langle\alpha\rangle}C{\rightarrow}{\bot}{\in}\Delta$, then ${\bot}{\in}\Delta$: a contradiction.
Thus, ${\lbrack\alpha\rbrack}\Delta{\in}{\mathcal T}$.
Consequently, ${\mathcal T}$ is nonempty.
Moreover, for all nonempty chains $(\Lambda_{i})_{i{\in}I}$ of elements of ${\mathcal T}$, $\bigcup\{\Lambda_{i}\ :\ i{\in}I\}$ is an element of ${\mathcal T}$.
Hence, by Zorn's Lemma, ${\mathcal T}$ possesses a maximal element $\Lambda$.
Thus, $\Lambda$ is a $\L$-theory such that ${\lbrack\alpha\rbrack}\Delta{\subseteq}\Lambda$ and ${\langle\alpha\rangle}\Lambda{\subseteq}\Delta$.
Consequently, it only remains to be proved that $\Lambda$ is proper and prime.
We claim that $\Lambda$ is proper.
If not, ${\bot}{\in}\Lambda$.
Since, ${\langle\alpha\rangle}\Lambda{\subseteq}\Delta$, then ${\langle\alpha\rangle}{\bot}{\in}\Delta$: a contradiction.
Hence, $\Lambda$ is proper.
We claim that $\Lambda$ is prime.
If not, there exists formulas $D,E$ such that $D{\vee}E{\in}\Lambda$, $D{\not\in}\Lambda$ and $E{\not\in}\Lambda$.
Thus, by the maximality of $\Lambda$ in ${\mathcal T}$, $\Lambda{+}D{\not\in}{\mathcal T}$ and $\Lambda{+}E{\not\in}{\mathcal T}$.
Consequently, there exists a formula $F$ such that $F{\in}\Lambda{+}D$ and ${\langle\alpha\rangle}F{\not\in}\Delta$ and there exists a formula $G$ such that $G{\in}\Lambda{+}E$ and ${\langle\alpha\rangle}G{\not\in}\Delta$.
Hence, $D{\rightarrow}F{\in}\Lambda$ and $E{\rightarrow}G{\in}\Lambda$.
Thus, $D{\vee}E{\rightarrow}F{\vee}G{\in}\Lambda$.
Since $D{\vee}E{\in}\Lambda$, then $F{\vee}G{\in}\Lambda$.
Since ${\langle\alpha\rangle}\Lambda{\subseteq}\Delta$, then ${\langle\alpha\rangle}(F{\vee}G){\in}\Delta$.
Consequently, ${\langle\alpha\rangle}F{\in}\Delta$, or ${\langle\alpha\rangle}G{\in}\Delta$.
Since ${\langle\alpha\rangle}F{\not\in}\Delta$, then ${\langle\alpha\rangle}G{\in}\Delta$: a contradiction.
Hence, $\Lambda$ is prime.
\medskip
\end{proof}
\begin{proposition}\label{proposition:A6:A7:canonical:frame}
If $\L$ contains $(\Axiom12)$ and $(\Axiom13)$ then $(W_{\L},{\leq_{\L}},{R_{\L}})$ is prestandard.
\end{proposition}
\begin{proof}
Suppose $\L$ contains $(\Axiom12)$ and $(\Axiom13)$.
Let $\alpha,\beta$ be groups and $\Gamma,\Delta{\in}W_{\L}$ be such that $\Gamma{R_{\L}(\alpha{\cup}\beta)}\Delta$.
Hence, $\lbrack\alpha{\cup}\beta\rbrack\Gamma{\subseteq}\Delta$ and $\langle\alpha{\cup}\beta\rangle\Delta{\subseteq}\Gamma$.
For the sake of the contradiction, suppose not $\Gamma{R_{\L}(\alpha)}\Delta$, or not $\Gamma{R_{\L}(\beta)}\Delta$.
Without loss of generality, suppose not $\Gamma{R_{\L}(\alpha)}\Delta$.
Thus, ${\lbrack\alpha\rbrack}\Gamma{\not\subseteq}\Delta$, or ${\langle\alpha\rangle}\Delta{\not\subseteq}\Gamma$.
In the former case, there exists a formula $A$ such that ${\lbrack\alpha\rbrack}A{\in}\Gamma$ and $A{\not\in}\Delta$.
Since $\L$ contains $(\Axiom12)$, then ${\lbrack\alpha{\cup}\beta\rbrack}A{\in}\Gamma$.
Since $\lbrack\alpha{\cup}\beta\rbrack\Gamma{\subseteq}\Delta$, then $A{\in}\Delta$: a contradiction.
In the latter case, there exists a formula $A$ such that $A{\in}\Delta$ and ${\langle\alpha\rangle}A{\not\in}\Gamma$.
Since $\L$ contains $(\Axiom13)$, then ${\langle\alpha{\cup}\beta\rangle}A{\not\in}\Gamma$.
Since $\langle\alpha{\cup}\beta\rangle\Delta{\subseteq}\Gamma$, then $A{\not\in}\Delta$: a contradiction.
\medskip
\end{proof}
Let $V_{\L}\ :\ \At{\longrightarrow}\wp(W_{\L})$ be the valuation on $(W_{\L},{\leq_{\L}},{R_{\L}})$ such that for all atoms $p$, $V_{\L}(p){=}\{\Gamma{\in}W_{\L}\ :\ p{\in}\Gamma\}$.
The valuation $V_{\L}\ :\ \At{\longrightarrow}\wp(W_{\L})$ on $(W_{\L},{\leq_{\L}},{R_{\L}})$ is called {\em canonical valuation of $\L$.}
The model $(W_{\L},{\leq_{\L}},{R_{\L}},V_{\L})$ is called {\em canonical model of $\L$.}
\begin{proposition}\label{proposition:truth:lemma}
For all formulas $A$ and for all $\Gamma{\in}W_{\L}$, $A{\in}\Gamma$ if and only if $\Gamma{\models}A$.
\end{proposition}
\begin{proof}
By induction on $A$.
\medskip
\end{proof}
\section{Soundness and completeness}\label{section:soundness:completeness}
In Proposition~\ref{proposition:soundness:completeness:sans:D}, we prove that $\L_{\all}$, $\L_{\doxastic}$, $\L_{\epistemic}$ and $\L_{\partition}$ are complete with respect to the relational semantics determined respectively by ${\mathcal C}_{\all}$, ${\mathcal C}_{\doxastic}$, ${\mathcal C}_{\epistemic}$ and ${\mathcal C}_{\partition}$.
\begin{proposition}[Soundness/completeness]\label{proposition:soundness:completeness:sans:D}
\begin{enumerate}
\item $\L_{\all}{=}\Log({\mathcal C}_{\all})$,
\item $\L_{\doxastic}{=}\Log({\mathcal C}_{\doxastic})$,
\item $\L_{\epistemic}{=}\Log({\mathcal C}_{\epistemic})$,
\item $\L_{\partition}{=}\Log({\mathcal C}_{\partition})$.
\end{enumerate}
\end{proposition}
\begin{proof}
By Propositions~\ref{proposition:validity:of:axioms}, \ref{proposition:validity:of:rules}, \ref{proposition:validity:of:specific:axioms:doxastic:frames}, \ref{proposition:validity:of:specific:axioms:epistemic:frames}, \ref{proposition:validity:of:up:down}, \ref{proposition:refsym:updown:partition:a}, \ref{proposition:A8:A9:canonical:frame} and~\ref{proposition:truth:lemma}.
Indeed, suppose for example that $\L_{\all}{\not=}\Log({\mathcal C}_{\all})$.
Hence, $\L_{\all}{\not\subseteq}\Log({\mathcal C}_{\all})$, or $\L_{\all}{\not\supseteq}\Log({\mathcal C}_{\all})$.
In the former case, there exists a formula $A$ such that $A{\in}\L_{\all}$ and ${\mathcal C}_{\all}{\not\models}A$.
This is in contradiction with Propositions~\ref{proposition:validity:of:axioms} and \ref{proposition:validity:of:rules}.
In the latter case, there exists a formula $A$ such that $A{\not\in}\L_{\all}$ and ${\mathcal C}_{\all}{\models}A$.
Thus, by Proposition~\ref{lemma:almost:completeness} there exists a prime $\L_{\all}$-theory $\Gamma$ such that $A{\not\in}\Gamma$.
Consequently, by Proposition~\ref{proposition:truth:lemma}, $\Gamma{\not\models}A$.
Hence, $(W_{\L},{\leq_{\L}},{R_{\L}},V_{\L}){\not\models}A$.
Thus, $(W_{\L},{\leq_{\L}},{R_{\L}}){\not\models}A$.
Consequently, ${\mathcal C}_{\all}{\not\models}A$: a contradiction.
\medskip
\end{proof}
In Proposition~\ref{soundness:completeness:distributed}, we prove that $\L_{\all}^{\Distrib}$, $\L_{\doxastic}^{\Distrib}$, $\L_{\epistemic}^{\Distrib}$ and $\L_{\partition}^{\Distrib}$ are complete with respect to the relational semantics determined respectively by ${\mathcal C}_{\all}^{\standard}$, ${\mathcal C}_{\doxastic}^{\standard}$, ${\mathcal C}_{\epistemic}^{\standard}$ and ${\mathcal C}_{\partition}^{\standard}$.
\begin{proposition}[Soundness/completeness]\label{soundness:completeness:distributed}
\begin{enumerate}
\item $\L_{\all}^{\Distrib}{=}\Log({\mathcal C}_{\all}^{\standard})$,
\item $\L_{\doxastic}^{\Distrib}{=}\Log({\mathcal C}_{\doxastic}^{\standard})$,
\item $\L_{\epistemic}^{\Distrib}{=}\Log({\mathcal C}_{\epistemic}^{\standard})$,
\item $\L_{\partition}^{\Distrib}{=}\Log({\mathcal C}_{\partition}^{\standard})$.
\end{enumerate}
\end{proposition}
\begin{proof}
Similar to the proof of Proposition~\ref{proposition:soundness:completeness:sans:D}, this time using Propositions~\ref{proposition:validity:of:axioms}, \ref{proposition:validity:of:rules}, \ref{proposition:validity:of:specific:axioms:doxastic:frames}, \ref{proposition:validity:of:specific:axioms:epistemic:frames}, \ref{proposition:validity:of:up:down}, \ref{proposition:validity:of:axioms:inclusion}, \ref{proposition:bounded:morphic:images:all}, \ref{proposition:bounded:morphic:images:doxastic}, \ref{proposition:bounded:morphic:images:epistemic}, \ref{proposition:bounded:morphic:images:partition}, \ref{proposition:refsym:updown:partition:b}, \ref{proposition:A8:A9:canonical:frame}, \ref{proposition:A6:A7:canonical:frame} and~\ref{proposition:truth:lemma}.
\medskip
\end{proof}
In Proposition~\ref{proposition:IELmoins:CE}, we prove that $\L_{\doxastic}$ and $\L_{\doxastic}^{\Distrib}$ can be considered as conservative extensions of the intuitionistic epistemic logic $\IEL^{-}$ introduced by Artemov and Protopopescu~\cite{Artemov:Protopopescu:2016}.
\begin{proposition}\label{proposition:IELmoins:CE}
For all diamond-free $\ArtemovProtopopescu$-formulas $A$, the following conditions are equivalent:
\begin{enumerate}
\item $A{\in}\L_{\doxastic}$,
\item $A{\in}\L_{\doxastic}^{\Distrib}$,
\item $\tau(A){\in}\IEL^{-}$.
\end{enumerate}
\end{proposition}
\begin{proof}
Let $A$ be a diamond-free $\ArtemovProtopopescu$-formula.
Since $\L_{\doxastic}{\subseteq}\L_{\doxastic}^{\Distrib}$, then it suffices to prove that $\mathbf{(2){\Rightarrow}(3)}$ and $\mathbf{(3){\Rightarrow}(1)}$.
\\
$\mathbf{(2){\Rightarrow}(3)}:$
Suppose $A{\in}\L_{\doxastic}^{\Distrib}$.
For the sake of the contradiction, suppose $\tau(A){\not\in}\IEL^{-}$.
Hence, there exists an $\IEL^{-}$-structure $(W,{\leq},{R})$, there exists a valuation $V\ :\ \At{\longrightarrow}
$\linebreak$
\wp(W)$ on $(W,{\leq},{R})$ and there exists $s{\in}W$ such that $(W,{\leq},{R},V),s{\not\models}\tau(A)$.
Let $W^{\prime}{=}
$\linebreak$
W$.
Let ${\leq^{\prime}}$ be the preorder on $W^{\prime}$ such that for all $t,u{\in}W^{\prime}$, $t{\leq^{\prime}}u$ if and only if $t{\leq}u$.
Let ${R^{\prime}}\ :\ {\wp^{\star}}(\Ag){\longrightarrow}\wp(W^{\prime}{\times}W^{\prime})$ be the function such that for all groups $\alpha$ and for all $t,u{\in}W^{\prime}$, $t{R^{\prime}(\alpha)}u$ if and only if $t{R}u$.
\begin{claim}
The frame $(W^{\prime},{\leq^{\prime}},{R^{\prime}})$ is doxastic and standard.
\end{claim}
\begin{proofclaim}
Left to the reader
\medskip
\end{proofclaim}
\\
Let $V^{\prime}\ :\ \At{\longrightarrow}\wp(W^{\prime})$ be the valuation on $(W^{\prime},{\leq^{\prime}},{R^{\prime}})$ such that for all atoms $p$, $V^{\prime}{p}{=}V(p)$.
\begin{claim}
For all $B{\in}\Sf(A)$ and for all $t{\in}W$, the following conditions are equivalent:
\begin{enumerate}
\item $(W,{\leq},{R},V),t{\models}\tau(B)$,
\item $(W^{\prime},{\leq^{\prime}},{R^{\prime}},V^{\prime}),t{\models}B$.
\end{enumerate}
\end{claim}
\begin{proofclaim}
By induction on $B$.
\medskip
\end{proofclaim}
\\
Since $(W,{\leq},{R},V),s{\not\models}\tau(A)$, then $(W^{\prime},{\leq^{\prime}},{R^{\prime}},V^{\prime}),s{\not\models}A$.
Thus, $(W^{\prime},{\leq^{\prime}},{R^{\prime}},V^{\prime}){\not\models}A$.
Consequently, $(W^{\prime},{\leq^{\prime}},{R^{\prime}}){\not\models}A$.
Since the frame $(W^{\prime},{\leq^{\prime}},{R^{\prime}})$ is doxastic and standard, then by Proposition~\ref{soundness:completeness:distributed}, $A{\not\in}\L_{\doxastic}^{\Distrib}$: a contradiction.
\\
$\mathbf{(3){\Rightarrow}(1)}:$
Suppose $\tau(A){\in}\IEL^{-}$.
For the sake of the contradiction, suppose $A{\not\in}\L_{\doxastic}$.
Hence, there exists a doxastic frame $(W,{\leq},{R})$, there exists a valuation $V\ :\ \At{\longrightarrow}
$\linebreak$
\wp(W)$ on $(W,{\leq},{R})$ and there exists $s{\in}W$ such that $(W,{\leq},{R},V),s{\not\models}A$.
Let $\alpha$ be a group such that for all groups $\beta$, if the modal operator $\lbrack\beta\rbrack$ occur in $A$ then $\alpha{=}\beta$.\footnote{Such a group exists, seeing that $A$ is diamond-free.}
Let $W^{\prime}{=}W$.
Let $\leq^{\prime}$ be the preorder on $W^{\prime}$ such that for all $t,u{\in}W^{\prime}$, $t{\leq^{\prime}}u$ if and only if $t{\leq}u$.
Let $R^{\prime}$ be the binary relation on $W^{\prime}$ such that for all $t,u{\in}W^{\prime}$, $t{R^{\prime}}u$ if and only if $t{\leq}{\circ}{R(\alpha)}u$.
\begin{claim}
The relational structure $(W^{\prime},{\leq^{\prime}},{R^{\prime}})$ is an $\IEL^{-}$-structure.
\end{claim}
\begin{proofclaim}
Left to the reader
\medskip
\end{proofclaim}
\\
Let $V^{\prime}\ :\ \At{\longrightarrow}\wp(W^{\prime})$ be the valuation on $(W^{\prime},{\leq^{\prime}},{R^{\prime}})$ such that for all atoms $p$, $V^{\prime}{p}{=}V(p)$.
\begin{claim}
For all $B{\in}\Sf(A)$ and for all $t{\in}W$, the following conditions are equivalent:
\begin{enumerate}
\item $(W,{\leq},{R},V),t{\models}B$,
\item $(W^{\prime},{\leq^{\prime}},{R^{\prime}},V^{\prime}),t{\models}\tau(B)$.
\end{enumerate}
\end{claim}
\begin{proofclaim}
By induction on $B$.
\medskip
\end{proofclaim}
\\
Since $(W,{\leq},{R},V),s{\not\models}A$, then $(W^{\prime},{\leq^{\prime}},{R^{\prime}},V^{\prime}),s{\not\models}\tau(A)$.
Thus, $(W^{\prime},{\leq^{\prime}},{R^{\prime}},V^{\prime}){\not\models}
$\linebreak$
\tau(A)$.
Consequently, $(W^{\prime},{\leq^{\prime}},{R^{\prime}}){\not\models}\tau(A)$.
Since the relational structure $(W^{\prime},{\leq^{\prime}},{R^{\prime}})$ is an $\IEL^{-}$-structure, then by~\cite[Theorem~$4.6$]{Artemov:Protopopescu:2016}, $\tau(A){\not\in}\IEL^{-}$: a contradiction.
\medskip
\end{proof}
In Proposition~\ref{proposition:IEL:CE}, we prove that $\L_{\epistemic}$ and $\L_{\epistemic}^{\Distrib}$ can be considered as conservative extensions of the intuitionistic epistemic logic $\IEL$ introduced by Artemov and Protopopescu~\cite{Artemov:Protopopescu:2016}.
\begin{proposition}\label{proposition:IEL:CE}
For all diamond-free $\ArtemovProtopopescu$-formulas $A$, the following conditions are equivalent:
\begin{enumerate}
\item $A{\in}\L_{\epistemic}$,
\item $A{\in}\L_{\epistemic}^{\Distrib}$,
\item $\tau(A){\in}\IEL$.
\end{enumerate}
\end{proposition}
\begin{proof}
Let $A$ be a diamond-free $\ArtemovProtopopescu$-formula.
Since $\L_{\epistemic}{\subseteq}\L_{\epistemic}^{\Distrib}$, then it suffices to prove that $\mathbf{(2){\Rightarrow}(3)}$ and $\mathbf{(3){\Rightarrow}(1)}$.
\\
$\mathbf{(2){\Rightarrow}(3)}:$
Suppose $A{\in}\L_{\epistemic}^{\Distrib}$.
For the sake of the contradiction, suppose $\tau(A){\not\in}\IEL$.
Hence, there exists an $\IEL$-structure $(W,{\leq},{R})$, there exists a valuation $V\ :\ \At{\longrightarrow}
$\linebreak$
\wp(W)$ on $(W,{\leq},{R})$ and there exists $s{\in}W$ such that $(W,{\leq},{R},V),s{\not\models}\tau(A)$.
The rest of the proof is done by imitating the argument developed in the ``$\mathbf{(2){\Rightarrow}(3)}$'' part of the proof of Proposition~\ref{proposition:IELmoins:CE}.\footnote{The reader may easily verify that since the relational structure $(W,{\leq},{R})$ is an $\IEL$-structure, then the frame $(W^{\prime},{\leq^{\prime}},{R^{\prime}})$ constructed in this argument is epistemic and standard.}
\\
$\mathbf{(3){\Rightarrow}(1)}:$
Suppose $\tau(A){\in}\IEL$.
For the sake of the contradiction, suppose $A{\not\in}\L_{\epistemic}$.
Hence, there exists an epistemic frame $(W,{\leq},{R})$, there exists a valuation $V\ :\ \At{\longrightarrow}
$\linebreak$
\wp(W)$ on $(W,{\leq},{R})$ and there exists $s{\in}W$ such that $(W,{\leq},{R},V),s{\not\models}A$.
The rest of the proof is done by imitating the argument developed in the ``$\mathbf{(3){\Rightarrow}(1)}$'' part of the proof of Proposition~\ref{proposition:IELmoins:CE}.\footnote{The reader may easily verify that since the frame $(W,{\leq},{R})$ is epistemic, then the relational structure $(W^{\prime},{\leq^{\prime}},{R^{\prime}})$ constructed in this argument is an $\IEL$-structure.}
\medskip
\end{proof}
\section{Conclusion}
There is a growing interest in the development of intuitionistic modal logics~\cite{Olivetti:2022,IMLA:2017,Stewart:et:al:2018}.
And some of them have their place in the family of logics for reasoning about knowledge.
In this article, we have shown that they have also their place in the family of logics for multi-agent systems.
Much remains to be done.
\\
\\
Firstly, one may evaluate the computability of $\L_{\all}$, $\L_{\doxastic}$, $\L_{\epistemic}$, $\L_{\partition}$, $\L_{\all}^{\Distrib}$, $\L_{\doxastic}^{\Distrib}$, $\L_{\epistemic}^{\Distrib}$ and $\L_{\partition}^{\Distrib}$.
We conjecture that the membership problems in these intuitionistic modal logics are decidable.
In order to fix this conjecture, one may use, for example, a technique based on the two-variable monadic guarded fragment, or selective filtration, or terminating sequent calculi~\cite{Alechina:Shkatov:2006,Dalmonte:et:al:2021,Grefe:1996,Iemhoff:2018,Lin:Ma:2019}.\footnote{The reader should be aware that it is by no means easy to determine whether such-or-such intuitionistic modal logic is decidable, witness the fact that the decidability of $\IS4$~---~the intuitionistic modal logic obtained by adding to the intuitionistic modal logics $\IK$ introduced by Fischer Servi~\cite{FischerServi:1984} the formulas ${\square}p{\rightarrow}p$, $p{\rightarrow}{\lozenge}p$, ${\square}p{\rightarrow}{\square}{\square}p$ and ${\lozenge}{\lozenge}p{\rightarrow}{\lozenge}p$~---~has only been proved recently~\cite{Girlando:et:al:2023}.}
\\
\\
Secondly, one may determine with respect to which classes of frames are sound and complete the least intuitionistic modal logic containing $\L_{\partition}$ and axioms
\begin{description}
\item[$(\Axiom14)$] ${\lbrack\alpha\rbrack}p{\rightarrow}{\lbrack\alpha\rbrack}{\lbrack\alpha\rbrack}p$,
\item[$(\Axiom15)$] ${\langle\alpha\rangle}{\langle\alpha\rangle}p{\rightarrow}{\langle\alpha\rangle}p$
\end{description}
and the least intuitionistic modal logic containing $\L_{\partition}^{\Distrib}$ and axioms~$(\Axiom14)$ and~$(\Axiom15)$.\footnote{Here, the reader should remind that as far as classical modal validity is concerned, the mono-agent versions of these formulas~---~${\square}p{\rightarrow}{\square}{\square}p$ and ${\lozenge}{\lozenge}p{\rightarrow}{\lozenge}p$, characterizing the property of positive introspection of knowledge~---~correspond to the elementary condition of transitivity.
See~\cite[Chapter~$4$]{Blackburn:et:al:2001} and~\cite[Chapter~$3$]{Chagrov:Zakharyaschev:1997}.}
We conjecture that the former intuitionistic modal logic is sound and complete with respect to the class of all partitions $(W,{\leq},{R})$ such that for all groups $\alpha$, ${R(\alpha)}{\circ}{\leq}{\circ}{{R(\alpha)}}{\subseteq}
$\linebreak$
{\leq}{\circ}{{R(\alpha)}}{\circ}{\leq}$ and ${R(\alpha)}{\circ}{\geq}{\circ}{{R(\alpha)}}{\subseteq}{\geq}{\circ}{{R(\alpha)}}{\circ}{\geq}$ and the latter intuitionistic modal logic is sound and complete with respect to the class of all standard partitions $(W,{\leq},{R})$ such that for all groups $\alpha$, ${R(\alpha)}{\circ}{\leq}{\circ}{{R(\alpha)}}{\subseteq}{\leq}{\circ}{{R(\alpha)}}{\circ}{\leq}$ and ${R(\alpha)}{\circ}{\geq}{\circ}{{R(\alpha)}}{\subseteq}{\geq}{\circ}{{R(\alpha)}}{\circ}{\geq}$.
\\
\\
Thirdly, one may determine with respect to which classes of frames are sound and complete the least intuitionistic modal logic containing $\L_{\partition}$ and axioms
\begin{description}
\item[$(\Axiom16)$] ${\langle\alpha\rangle}p{\rightarrow}{\lbrack\alpha\rbrack}{\langle\alpha\rangle}p$,
\item[$(\Axiom17)$] ${\langle\alpha\rangle}{\lbrack\alpha\rbrack}p{\rightarrow}{\lbrack\alpha\rbrack}p$
\end{description}
and the least intuitionistic modal logic containing $\L_{\partition}^{\Distrib}$ and axioms~$(\Axiom16)$ and~$(\Axiom17)$.\footnote{Here, the reader should remind that as far as classical modal validity is concerned, the mono-agent versions of these formulas~---~${\lozenge}p{\rightarrow}{\square}{\lozenge}p$ and ${\lozenge}{\square}p{\rightarrow}{\square}p$, characterizing the property of negative introspection of knowledge~---~correspond to the elementary condition of Euclideanity.
See~\cite[Chapter~$4$]{Blackburn:et:al:2001} and~\cite[Chapter~$3$]{Chagrov:Zakharyaschev:1997}.}
We conjecture that the former intuitionistic modal logic is sound and complete with respect to the class of all partitions $(W,{\leq},{R})$ such that for all groups $\alpha$, ${R(\alpha)^{-1}}{\circ}{\leq}{\circ}{{R(\alpha)}}
$\linebreak$
{\subseteq}{\leq}{\circ}{{R(\alpha)^{-1}}}{\circ}{\leq}$ and ${R(\alpha)^{-1}}{\circ}{\geq}{\circ}{{R(\alpha)}}{\subseteq}{\geq}{\circ}{{R(\alpha)^{-1}}}{\circ}{\geq}$ and the latter intuitionistic modal logic is sound and complete with respect to the class of all standard partitions $(W,{\leq},{R})$ such that for all groups $\alpha$, ${R(\alpha)^{-1}}{\circ}{\leq}{\circ}{{R(\alpha)}}{\subseteq}{\leq}{\circ}{{R(\alpha)^{-1}}}{\circ}{\leq}$ and ${R(\alpha)^{-1}}{\circ}
$\linebreak$
{\geq}{\circ}{{R(\alpha)}}{\subseteq}{\geq}{\circ}{{R(\alpha)^{-1}}}{\circ}{\geq}$.
\\
\\
Fourthly, one may include in our language the notion of common knowledge as done within the context of another intuitionistic modal logic by J\"ager and Marti~\cite{Jager:Marti:2016}.
This means for all groups $\alpha$, to include in our language the modal operators $\lbrack\alpha^{\star}\rbrack$ and $\langle\alpha^{\star}\rangle$, ${\lbrack\alpha^{\star}\rbrack}A$ and ${\langle\alpha^{\star}\rangle}A$ being respectively read for all formulas $A$, ``$A$ is consequence of $\alpha$'s common knowledge'' and ``$A$ is compatible with $\alpha$'s common knowledge''.
%
%
\\
\\
Fifthly, one may include in our language the notion of public announcement as done within the context of other intuitionistic modal logics by Ma {\em et al.}~\cite{Ma:Palmigiano:Sadrzadeh:2014} and Murai and Sano~\cite{Murai:Sano:2022}.
This means for all formulas $A$, to include in our language the modal operators $\lbrack A\rbrack$ and $\langle A\rangle$, ${\lbrack A\rbrack}B$ and ${\langle A\rangle}B$ being respectively read for all formulas $B$, ``$B$ is consequence of the public announcement of $A$'' and ``$A$ is compatible with the public announcement of $A$''.
%
%
%
%
\section*{Acknowledgements}
We wish to thank Han Gao (Aix-Marseille University), \c{C}i\u{g}dem Gencer (Toulouse University), Zhe Lin (Xiamen University), Nicola Olivetti (Aix-Marseille University) and Vladimir Sotirov (Bulgarian Academy of Sciences) for their valuable remarks.
%
%
%
%
\bibliographystyle{named}
\end{document}